\documentclass[11pt]{article}

\usepackage{amssymb,amsmath,amsthm,amsfonts, bm,array}
\usepackage{fullpage,nicefrac,comment,graphicx}
\usepackage[ruled,vlined]{algorithm2e}

\usepackage{xcolor}
\usepackage{color}
\usepackage{hyperref}
\usepackage{tikz}
\usetikzlibrary{arrows,calc}

\usepackage{subfigure}
\newcommand{\A}{\ensuremath{\mathcal{A}}}

\newcommand{\R}{{\mathbb R}}

\newcommand{\matA}{\mathbf{A}}

\newcommand{\E}[1]{\mathbb{E}\left[ #1\right]}

\newcommand{\onenorm}[1]{\left\|#1\right\|_1}
\newcommand{\zeronorm}[1]{\left\|#1\right\|_0}
\newcommand{\twonorm}[1]{\left\|#1\right\|_2}
\newcommand{\inftynorm}[1]{\left\|#1\right\|_\infty}

\newcommand{\supp}{\mathrm{supp}}

\newcommand{\argmin}{\mathrm{argmin}}

\newcommand{\sgn}{\mathrm{sign}}

\newcommand{\ip}[2]{\ensuremath{\left\langle #1,#2\right\rangle}}
\newcommand{\inset}[1]{\left\{#1\right\}}
\newcommand{\inparen}[1]{\left(#1\right)}

\newtheorem{theorem}{Theorem} 
 
\newtheorem{definition}[theorem]{Definition}

\newtheorem{cor}[theorem]{Corollary} 

\newtheorem{remark}{Remark}

\newcommand{\defby}{\overset{\mathrm{\scriptscriptstyle{def}}}{=}}

\newcommand{\sign}{\mathrm{sign}}

\newcommand{\vect}[1]{\bm{#1}}
\newcommand{\mat}[1]{\bm{#1}}
\def \< {\langle}
\def \> {\rangle}

\def \x {\vect{x}}
\def \u {\vect{u}}
\def \w {\vect{w}}
\def \y {\vect{y}}
\def \v {\vect{v}}

\def \A {\mat{A}}
\def \z {\vect{z}}
\def \a {\vect{a}}
\def \t {\vect{t}}
\def \e {\vect{e}}

\begin{document}

\title{\bf Exponential decay of reconstruction error\\ from binary measurements of sparse signals}

\author{Richard Baraniuk$^r$, Simon Foucart$^g$, Deanna Needell$^c$, Yaniv Plan$^b$, Mary Wootters$^m$%
\thanks{Authors are listed in alphabetical order.}
\thanks{This research was performed as part of the AIM SQuaRE program.  
In addition,
Baraniuk is partially supported by NSF grant number CCF-0926127 and ARO MURI grant number W911NF-09-1-0383,
Foucart by NSF grant number DMS--1120622,
Needell by Simons Foundation Collaboration grant number 274305, Alfred P. Sloan Fellowship and NSF Career grant number 1348721,
Plan by an NSF Postdoctoral Research Fellowship grant number 1103909, 
Wootters by a Rackham Predoctoral Fellowship and by the Simons Institute for the Theory of Computing.}
\\[3mm]
\small $^r$Department of Electrical and Computer Engineering,
Rice University,  \\[-1mm]
\small 6100 Main Street, Houston, TX 77005 USA. Email: richb@rice.edu\\[1mm]
\small $^g$Department of Mathematics, University of Georgia \\[-1mm]
\small 321C Boyd Building, Athens, GA 30602 USA. Email: foucart@math.uga.edu\\[1mm]
\small $^c$Department of Mathematical Sciences, Claremont McKenna College, \\[-1mm]
\small 850 Columbia Ave., Claremont, CA 91711, USA. Email: dneedell@cmc.edu \\[1mm]
\small $^b$Department of Mathematics, University of British Columbia \\[-1mm]
\small 1984 Mathematics Road, Vancouver, B.C. Canada V6T 1Z2. Email: yaniv@math.ubc.ca\\[1mm]
\small $^m$Department of Mathematics, University of Michigan \\[-1mm]
\small  530 Church Street, Ann Arbor, MI 48109, USA. Email: wootters@umich.edu
}

\maketitle

\begin{abstract}
Binary measurements arise naturally in a variety of statistical and engineering applications.  
They may be inherent to the problem---e.g., in determining the relationship between genetics and the presence or absence of a disease---or they may be a result of extreme quantization.  
A recent influx of literature has suggested that using prior signal information can greatly improve the ability to reconstruct a signal from binary measurements.  
This is exemplified by \textit{one-bit compressed sensing}, which takes the compressed sensing model but assumes that only the sign of each measurement is retained.   
It has recently been shown that the number of one-bit measurements required for signal estimation mirrors that of unquantized compressed sensing.  Indeed, $s$-sparse signals in $\R^n$ can be estimated (up to normalization) from $\Omega(s \log (n/s))$ one-bit measurements.  
Nevertheless, controlling the precise accuracy of the error estimate remains an open challenge.  
In this paper, we focus on optimizing the decay of the error as a function of the oversampling factor $\lambda := m/(s \log(n/s))$,
where $m$ is the number of measurements.  
It is known that the error in reconstructing sparse signals from standard one-bit measurements is bounded below by $\Omega(\lambda^{-1})$.   
Without adjusting the measurement procedure, reducing this polynomial error decay rate is impossible.  
However, we show that an adaptive choice of the thresholds used for quantization  may lower the error rate to $e^{-\Omega(\lambda)}$.  
This improves upon guarantees for other methods of adaptive thresholding as proposed in Sigma-Delta quantization.  
We develop a general recursive strategy to achieve this exponential decay and two specific polynomial-time algorithms which fall into this framework, one based on convex programming and one on hard thresholding.  
This work is inspired by the one-bit compressed sensing model, in which the engineer controls the measurement procedure.  
Nevertheless, the principle is extendable to signal reconstruction problems in a variety of binary statistical models as well as statistical estimation problems like logistic regression.     
    
\end{abstract}


\section{Introduction}
\label{sec:intro}

Many practical acquisition devices in signal processing and algorithms in machine learning use a small number of linear measurements to represent a high-dimensional signal. 
Compressed sensing is a technology which takes advantage of the fact that, for some interesting classes of signals, one can use far fewer measurements than dictated by traditional Nyquist sampling paradigm.  
In this setting, one obtains $m$ linear measurements of a signal $\x \in \R^n$ of the form
$$
y_i = \ip{\a_i}{\x},
\qquad i=1,\ldots,m.
$$
Written concisely, one obtains the measurement vector $\y = \matA\x$, where $\matA \in \R^{m \times n}$ is the matrix with rows $\a_1,\ldots,\a_m$.  From these (or even from corrupted measurements $\y = \matA\x + \e$), one wishes to recover the signal $\x$.  To make this problem well-posed, one must exploit a priori information on the signal $\x$, for example that it is \em $s$-sparse\em, i.e.,
$$
\|\x\|_0 \defby |\supp(\x)| = s \ll n,
$$
or is well-approximated by an $s$-sparse signal.  
After a great deal of research activity in the past decade
(see the website \cite{DSPweb} or the references in the monographs \cite{eldar2012compressed,foucart2013}),
it is now well known that when $\matA$ consists of, say, independent standard normal entries, 
one can, with high probability, recover all $s$-sparse vectors ${\x}$ from the $m \approx s\log(n/s)$ linear measurements $y_i = \ip{\a_i}{\x}$, $i=1,\ldots,m$.

However, in practice, the compressive measurements $\ip{\a_i}{\x}$ must be quantized:
one actually observes $\y = Q(\matA \x)$, where the map $Q: \R^m \to \mathcal{A}^m$ is a quantizer that acts entrywise by mapping each real-valued measurement to a discrete quantization alphabet $\mathcal{A}$.  
This type of quantization with an alphabet $\mathcal{A}$ consisting of only two elements was introduced in the compressed sensing setting by \cite{Boufounos2008} and dubbed \em one-bit compressed sensing \em.  
In this work, we focus on this one-bit approach and seek quantization schemes $Q$ and reconstruction algorithms $\Delta$ so that $\hat{\x} = \Delta( Q( \matA \x) )$ is a good approximation to $\x$.
In particular, we are interested in the trade-off between the error of the approximation and the 
  \em oversampling factor \em
\[ \lambda \defby \frac{m}{s\log(n/s)}.\]

\subsection{Motivation and previous work}
  
\sloppy
The most natural quantization method is \em Memoryless Scalar Quantization \em (MSQ),
where each entry of $\y = \matA \x$ is rounded to the nearest element of some quantization alphabet $\mathcal{A}$.  
If $\mathcal{A}= \delta \mathbb{Z}$ for some suitably small $\delta > 0$, then this rounding error can be modeled as an additive measurement error~\cite{Dai2009}, and the recovery algorithm can be fine-tuned to this particular situation~\cite{jacques2011dequantizing}.  
In the one-bit case, however, the quantization alphabet is $\mathcal{A} = \inset{\pm 1}$ and the quantized measurements take the form $\y = \sgn (\A \x)$,
meaning that $\sgn$\footnote{We define $\sgn(0)=1$.}
acts entrywise as
\[y_i = Q_{\rm MSQ}(\ip{\a_i}{\x}) = \sgn( \ip{\a_i}{\x}),
\qquad i=1,\ldots,m. \]
One-bit compressed sensing was introduced in~\cite{Boufounos2008}, 
and it has generated a considerable amount of work since then,
see~\cite{DSPweb} for a growing list of literature in this area.  
Several efficient recovery algorithms have been proposed,
based on linear programming~\cite{pv-1-bit,pv-noisy-1bit, gnjn2013}
and on modifications of iterative hard thresholding~\cite{Jacques2011,jacques2013quantized}.
As shown in \cite{Jacques2011},  with high probability 
one can perform the reconstruction from one-bit measurements with error
\[ \twonorm{ \x - \hat{\x} } \lesssim \frac{1}{\lambda} 
\qquad \text{ for all } \x \in \Sigma_s' := \{ \v \in \R^n \ :\ \|\v\|_0 \leq s, \|\v\|_2 = 1\}. \]
In other words, a  uniform $\ell_2$-reconstruction error of at most $\gamma>0$ can be achieved with 
$m \asymp \gamma^{-1} s\log(n/s)$ one-bit measurements.

\sloppy
Despite the dimension reduction from $n$ to $s\log(n/s)$, MSQ presents substantial limitations~\cite{Jacques2011, goyal1998quantized}.  
 Precisely, according to~\cite{goyal1998quantized}, even if the support of $\x \in \Sigma_s'$ is known, 
  the best recovery algorithm $\Delta_{\rm opt}$ must obey
\begin{equation}\label{eq:lowerbound}
 \twonorm{ \x - \Delta_{\rm opt}( Q_{\rm MSQ}(\matA \x) ) } \gtrsim \frac{1}{\lambda}
\end{equation}
up to a logarithmic factor.
An intuition for the limited accuracy of MSQ is given in Figure \ref{fig:tesselate}.

\begin{figure}
\begin{center}
\begin{tikzpicture}[scale=1.5]

\draw[fill=gray!50,gray!50] (-.2,-1) ellipse (1cm and .1cm);
\draw[fill=gray!70,gray!70] (-.1,-1) ellipse (.6cm and .06cm);
\draw[very thick,fill=white] (0,0) circle (1cm);
\begin{scope}
\foreach \x in {-90,130}
{
\clip[rotate=\x] (0,0) ellipse (.87cm and .46cm);
}
\clip (-.4,.3) rectangle (.3,1);
\fill[blue!60!black!30] (0,0) circle (1cm);
\end{scope}
\begin{scope}
\foreach \x in {200}
{
\clip[rotate=\x] (0,0) ellipse (.9cm and .45cm);
}
\clip (-.4,.3) rectangle (.3,1);
\fill[white] (0,0) circle (1cm);
\end{scope}

\foreach \x in {0,80,-90,130,200}
{
\begin{scope}[rotate=\x]
 \pgfsetlinewidth{1pt}
    \pgfmoveto{\pgfpoint{1cm}{0cm}}
    \pgfpatharcto{1.1cm}{.8cm}{0}{0}{0}{\pgfpoint{-1cm}{-0cm}}\pgfusepath{stroke};
\end{scope}
}
\node at (1.2,-.8) {$S^{n-1}$};
\node (x) at (-1.5,.6) {$\x$};
\node[circle, fill=black, draw, scale=.2](xpt) at (-.2,.6){};
\draw[->] (x) to (xpt) {};
\end{tikzpicture}
\end{center}
\caption{\em Geometric interpretation of one-bit compressed sensing.  \em Each quantized measurement reveals which side of a hyperplane (or great circle, when restricted to the sphere) the signal $\x$ lies on.  After several measurements, we know that $\x$ lies in one unique region.  However, if the measurements are non-adaptive, then as the region of interest becomes smaller, it becomes less and less likely that the next measurement yields any new information about $\x$.}
\label{fig:tesselate}
\end{figure}
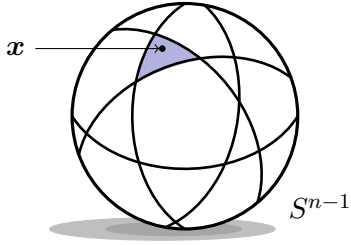

Alternative quantization schemes have been developed to overcome this drawback.  
For a specific signal model and reconstruction algorithm, \cite{Sun2009} obtained the optimal quantization scheme, but more general quantization schemes remain open.

Recently, Sigma-Delta quantization schemes have also been proposed as a more general quantization model~\cite{gunturk2010,krahmer2013sigma}.  
These works show that, with high probability on measurement matrices with independent subgaussian entries, 
$r$-th order Sigma-Delta quantization can be applied to the standard compressed sensing problem to achieve,
for any $\alpha \in (0,1)$, the reconstruction error 
\begin{equation}
\twonorm{ \x - \hat{\x} } \lesssim_r \lambda^{-\alpha(r - 1/2)} 
\label{eq:xx1}
\end{equation}
with a number of measurements
\[ m \approx s \inparen{\log(n/s)}^{1/(1-\alpha)}. \]
For suitable choices of $\alpha$ and $r$, the guarantee (\ref{eq:xx1}) overcomes the limitation~\eqref{eq:lowerbound},
but it is still polynomial in~$\lambda$.
This leads us to ask whether an exponential dependence can be achieved.

\subsection{Our contributions}

In this work, we focus on improving the trade-off between the error $\twonorm{\x - \hat{\x}}$ and the oversampling factor $\lambda$.
To the best of our knowledge, all quantized compressed sensing schemes obtain guarantees of the form
\begin{equation}\label{eq:bestknown}
 \twonorm{\x - \hat{\x}} \lesssim \lambda^{-c} \qquad \text{ for all $\x \in \Sigma_s'$} 
 \end{equation}
with some constant $c>0$.
We develop one-bit quantizers $Q: \R^m \to \inset{\pm 1}$, 
coupled with two efficient recovery algorithms $\Delta: \inset{\pm 1} \to \R^m$
that yield the reconstruction guarantee
\begin{equation}\label{eq:guarantee}
 \twonorm{\x -\Delta(Q(\matA \x))} \leq \exp(-\Omega(\lambda)) \qquad \text{ for all $\x \in \Sigma_s'$}.
\end{equation}
It is not hard to see that the dependence on $\lambda$ in (\ref{eq:guarantee}) is optimal, 
since any method of quantizing measurements that provides the reconstruction guarantee $\twonorm{\x - \hat{\x}} \leq \gamma$ must use at least  $\log_2 \mathcal{N}(\Sigma_s', \gamma) \geq s\log_2(1/\gamma)$ bits, where $\mathcal{N}(\cdot)$ denotes the covering number.

\subsubsection{Adaptive measurement model}

A key element of our approach is that the quantizers are {\em adaptive} to previous measurements of the signal in a manner similar to Sigma-Delta quantization~\cite{gunturk2010}.  
In particular, the measurement matrix $\matA \in \R^{m \times n}$ is assumed to have independent standard normal entries
and the quantized measurements take the form of thresholded signs, i.e., 
\begin{equation}\label{eq:query}
 y_i = \sgn(\ip{ \a_i }{\x} - \tau_i) = \left\{
     \begin{array}{lrr}
       1 & \text{if} & \ip{ \a_i }{\x} \geq \tau_i,\\
       -1 & \text{if} & \ip{ \a_i }{\x} < \tau_i.
     \end{array}
   \right.
\end{equation}
Such measurements are readily implementable in hardware, and they retain the simplicity and storage benefits of the one-bit compressed sensing model.  However, as we will show, this model is much more powerful in the sense that it permits optimal guarantees of the form \eqref{eq:guarantee}, which are impossible with standard MSQ one-bit quantization.
As in the Sigma-Delta quantization approach, we allow the quantizer to be adaptive, 
meaning that the quantization threshold $\tau_i$ of the $i$th entry may depend on the $1$st, $2$nd, $\ldots$, $(i-1)$st quantized measurements.
In the context of \eqref{eq:query}, this means that the thresholds $\tau_i$ will be chosen adaptively, resulting in a feedback loop as depicted in Figure \ref{fig:loop}. 
The thresholds $\tau_i$ can also be interpreted as an additive {\em dither}, which is oft-used in the theory and practice of analog-to-digital conversion.

\begin{figure}
\begin{center}
\begin{tikzpicture}[scale=3]
\draw[->] (0.6,0) -- (1,0);
\draw[->] (2.3,0) -- (2.5,0);
\draw[->] (2.7,0) -- (2.95,0);
\draw[->] (4.2,0) -- (4.6,0);
\draw[dashed, ->] (2.6,-0.4) -- (2.6,-0.1);
\draw[dashed] (4, -0.1) -- (4, -0.5);
\draw[dashed, ->] (4, -0.5) -- ( 2.8, -0.5);
\draw (1.1,-0.2) rectangle (1.7,0.2);
\draw (3.1,-0.2) rectangle (3.7,0.2);
\node at (1.4,0) {$\matA$};
\node at (2.6, 0) {$\bigoplus$};
\node at (3.4, 0) {quantize};
\node at (0.4, 0) {$\x\in\mathbb{R}^n$};
\node at (2, 0) {$\matA\x\in\mathbb{R}^m$};  
\node at (4,0) {$\y\in\mathbb{R}^m$};
\node at (2.6, -0.5) {$\tau\in\mathbb{R}^m$};
\end{tikzpicture}
\end{center}
\caption{Our closed-loop feedback system for binary measurements.}
\label{fig:loop}
\end{figure}
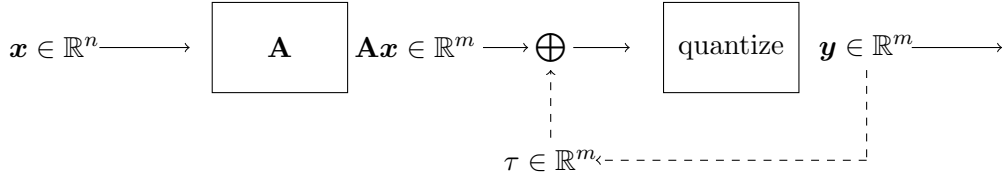

In contrast to Sigma-Delta quantization, the feedback loop involves the calculation of the quantization threshold.  This is the concession made to arrive at exponentially decaying error rates.  
It is an interesting open problem to determine low-memory quantization methods with such error rates that do not require such a calculation.

\subsubsection{Overview of our main result}

Our main result is that there is a recovery algorithm
 using measurements of the form \eqref{eq:query} 
and providing a guarantee of the form \eqref{eq:guarantee}.
For clarity of exposition, we overview a simplified version of our main result below.  The full result is stated in Section \ref{sec:strategy}.  
\begin{theorem}[Main theorem, simplified version]\label{thm:informal}
Let $Q$ and $\Delta$ be the quantization and recovery algorithms given below in Algorithms \ref{alg:simpleQ} and \ref{alg:simpleDelta}, respectively.  Suppose that $\A \in \R^{m \times n}$ and $\tau \in \R^{m}$ have independent standard normal entries.  Then, with probability at least $C\lambda \exp(-cs\log(n/s))$ over the choice of $\A$ and $\tau$, for all $\x \in B_2^n$ with $\|\x\|_0 \leq s$,
\[ \twonorm{ \x - \Delta( Q( \A\x, \A, \tau))  } \leq \exp( -\Omega(\lambda)),  \qquad \text{where} \qquad \lambda = \frac{m}{s\log(n/s)}. \]
\end{theorem}

The quantization algorithm works iteratively.  First, a small batch of measurements are quantized in a memoryless fashion.  From this first batch, one gains a very rough estimate of $\x$ (called $\x_1$).  The next batch of measurements are quantized with a focus on encoding the difference between $\x$ and $\x_1$, and so on.  Thus, the trap depicted in Figure \ref{fig:tesselate} is avoided; each hyperplane is translated with an appropriate dither, with the aim of cutting the size of the feasible region.
The recovery algorithm also works iteratively and its iterates are in fact intertwined with the iterates of the quantization algorithm.
We artificially separate the two algorithms  below.

\begin{algorithm}[t]
\KwIn{ Linear measurements $\matA\x \in \R^m$;
measurement matrix $\A \in \R^{m \times n}$;  
sparsity parameter $s$;
thresholds $\tau \in~\R^m$;
parameter $q \geq Cs\log(n/s)$ for the size of batches.}
\KwOut{Quantized measurements $\y \in \{\pm 1\}^m$. }

$T \gets \left\lfloor \frac{m}{q} \right\rfloor$

Partition $\A$ and $\tau$ into $T$ blocks $\A^{(1)}, \ldots, \A^{(T)} \in \R^{q \times n}$ and $\tau^{(1)},\ldots.\tau^{(T)} \in \R^q$.

$\x_0 \gets {\bf 0}$

\For{$t=1,\ldots,T$}
{
	$\mathbf{\sigma}^{(t)} \gets \matA^{(t)} \x_{t-1}$	
	
	$\y^{(t)} \gets \sign(\A^{(t)} \x - 2^{2-t}\mathbf{\tau}^{(t)} - \mathbf{\sigma}^{(t)})$
	
	$\z_t \gets \argmin \|\z\|_1 \qquad \text{subject to} \quad \|\z\|_2 \leq 2^{2-t}, \ y^{(t)}_i \left(\ip{ \a^{(t)}_i }{\z} - 2^{2-t} \tau^{(t)}_i \right) \geq 0 \quad  \text{for all } i$
	\tcp{$\z_t$ is an approximation of $\x - \x_{t-1}$}
	
	$\x_t \gets H_s( \x_{t-1} + \z_t )$ \tcp{$H_s$ keeps $s$ largest (in magnitude) entries and zeroes out the rest}  
}
\Return{$\y^{(t)} \mbox{ \rm for } t=1,\ldots,T$ }\tcp{Notice that we discard $\mathbf{\sigma}^{(t)}$}
\caption{Adaptive quantization}\label{alg:simpleQ}
\end{algorithm}

\begin{algorithm}[t]
\KwIn{ Quantized measurements $\y \in \{\pm 1\}^m$;  
measurement matrix $\matA$;
sparsity parameter $s$;
thresholds $\mathbf{\tau} \in \R^m$;
size of batches $q$.} 
\KwOut{  Approximation $\hat{\x} \in \R^n$. }

$T \gets \left\lfloor \frac{m}{q} \right\rfloor$

Partition $\A$ and $\tau$ into $T$ blocks $\A^{(1)}, \ldots, \A^{(T)} \in \R^{q \times n}$ and $\tau^{(1)},\ldots.\tau^{(T)} \in \R^q$.

$x_0 \gets {\bf 0}$

\For{ $t=1,\ldots,T$ }
{
	$\z_t \gets \argmin \|\z\|_1 \qquad\text{subject to} \quad\|\z\|_2 \leq 2^{2-t}, \ y^{(t)}_i \left(\ip{ \a^{(t)}_i }{\z} - 2^{2-t}\tau^{(t)}_i \right) \geq 0 \quad \text{for all } i$
	
	$\x_t = H_s( \x_{t-1} + \z_t )$ 
	}
\Return{$\x_T$}
\caption{Recovery}
\label{alg:simpleDelta}
\end{algorithm}

Note that we present Algorithms \ref{alg:simpleQ} and \ref{alg:simpleDelta} at this point mainly because they are the simplest to state.  Below we will provide a more general framework for algorithms that satisfy the guarantees of Theorem \ref{thm:informal} and develop a second set of algorithms with computational advantages.

\subsubsection{Robustness}

Our algorithms are robust to two different kinds of measurement corruption.
First, they allow for perturbed linear measurements of the form $\ip{\a_i}{\x} + e_i$ for an error vector $\e \in \R^m$ with bounded $\ell_\infty$-norm.  Second they allow for post-quantization sign flips, recorded as a vector $\mathbf{f} \in \{\pm 1\}^m$.

Formally, the measurements take the form
\begin{equation}\label{eq:noisyquery}
 y_i = f_i \, \sgn(\ip{ \a_i }{\x} - \tau_i + e_i), \qquad i=1,\ldots,m.
\end{equation}
It is known that for inaccurate measurements with pre-quantization noise on the same order of magnitude as the signal, 
even unquantized compressed sensing algorithms must obey a lower bound of the form \eqref{eq:lowerbound} \cite{candes2013well}.  
Our algorithms respect this reality and exhibit exponentially fast convergence until the estimate hits the ``noise floor''---that is, 
until the error $\twonorm{\x - \hat{\x}}$ is on the order of $\|\e\|_\infty$.

Table \ref{tab:algs} summarizes the various noise models, adaptive threshold calculations, and algorithms we develop and study below. 

\begin{table}[t]
\begin{center}
\caption{
Summary of the noise models, adaptive threshold calculations, and algorithms considered.  See Section~\ref{sec:magnitude recovery} for further discussion of the trade-offs between the two algorithms.}
\label{tab:algs}
\smallskip
\begin{tabular}{|>{\centering\arraybackslash}m{2in}|>{\centering\arraybackslash}m{2in}|>{\centering\arraybackslash}m{2in}|}
\hline
{\bf Noise model} & {\bf Threshold algorithm} & {\bf Recovery algorithm} \\ \hline
Additive error $e_i$ in \eqref{eq:noisyquery} & Algorithm~\ref{alg:Q}, instantiated by Algorithm~\ref{alg:T0socp} & Convex programming: Algorithm~\ref{alg:Delta}, instantiated by Algorithm~\ref{alg:Delta0socp} 
 \\ \hline
Additive error $e_i$ and sign flips $f_i$ in \eqref{eq:noisyquery} & Algorithm~\ref{alg:Q}, instantiated by Algorithm~\ref{alg:T0ht} & Iterative hard thresholding: Algorithm~\ref{alg:Delta}, instantiated by Algorithm~\ref{alg:Delta0ht}\\\hline
\end{tabular}
\end{center}
\end{table}

\subsubsection{Relationship to binary regression}

Our one-bit adaptive quantization and reconstruction algorithms are more broadly applicable to a certain kind of statistical classification problem related to sparse binary regression, and in particular sparse logistic and probit regression.
These techniques are often used to explain statistical data in which the response variable is binary. 
In regression, it is common to assume that the data $\{z_i\}$ is generated according to the \textit{generalized linear model}, where
$z_i \in \{0,1\}$ is a Bernoulli random variable satisfying
\begin{equation}\label{eq:gen lin}
\E {z_i} = f(\langle \a_i, \x \rangle)
\end{equation}
for some function $f : \R \rightarrow [0,1]$.
The generalized linear model is equivalent to the noisy one-bit compressed sensing model when the measurements $y_i = 2 z_i - 1 \in \{\pm 1\}$ and
\[
P(y_i = 1)  =: f(\langle \a_i, \x \rangle), 
\]
or equivalently, when 
\[
y_i = \sign(\langle \a_i, \x \rangle + e_i)
\]
with $f(t) := P(e_i \geq -t)$.
In summary, one-bit compressed sensing is equivalent to binary regression as long as $f$ is the cumulative distribution function (CDF) of the noise variable $e_i$.
The most commonly used CDFs in binary regression are the inverse logistic link function $f(t) = \frac{1}{1+e^{t}}$ in logistic regression and the inverse probit link function $f(t) = \int^t_{-\infty} \mathcal{N}(s) \mathrm{d}s$ in probit regression. 
These cases correspond to the noise variable $e_i$ being logistic and Gaussian distributed, respectively.

The new twist here is that the quantization thresholds are selected adaptively; see Section~\ref{sec:related} for some examples.
Specifically, our adaptive threshold measurement model is equivalent to the adaptive binary regression model
\[
y_i = \sign (\langle \a_i, \x \rangle + e_i - \tau_i )
\]
with
\[
P(y_i = 1) = P(e_i - \tau_i >= -t) = f( t - \tau_i).
\]
The effect of $\tau_i$ in this adaptive binary regression is equivalent to an offset term added to all measurements $y_i$.
Standard binary regression corresponds to the special case with $\tau_i = 0$.

\subsection{Organization}
In Section~\ref{sec:magnitude recovery}, we introduce two methods to recover not only the direction, but also the magnitude, of a signal from one-bit compressed sensing measurements of the form \eqref{eq:noisyquery}.  
These methods may be of independent interest (in one-bit compressed sensing, only the direction can be recovered), but they do not exhibit the exponential decay in the error that we seek.  In Section~\ref{sec:strategy}, we will show how to use these schemes as building blocks to obtain \eqref{eq:guarantee}.
The proofs of all of our results are given in Section \ref{sec:proofs}.
In Section~\ref{sec:exps}, we present some numerical results for the new algorithms.
We conclude in Section \ref{sec:discussion} with a brief summary. 

\subsection{Notation}
Throughout the paper,
we use the standard notation $\|\v\|_2 = \sqrt{\sum_i v_i^2}$ for the $\ell_2$-norm of a vector $\v \in \R^n$,
$\|\v\|_1 = \sum_i |v_i|$ for its $\ell_1$-norm,
and $\|\v\|_0$ for its number of nonzero entries.
A vector $\v$ is called $s$-sparse if $\|\v\|_0 \le s$
and effectively $s$-sparse if $\|\v\|_1 \le \sqrt{s} \|\v\|_2$.
We write $H_s(\v)$ to represent the vector in $\R^n$ agreeing with $\v$ on the index set of largest $s$ entries of $\v$ (in magnitude) and with the zero vector elsewhere.
We use a prime to indicate $\ell_2$-normalization,
so that $H_s'(\v)$ is defined as $H_s'(\v) := H_s(\v)/\twonorm{H_s(\v)}$.
The set
$\Sigma_s := \{ \v \in \R^n: \|\v\|_0 \le s \}$
of $s$-sparse vectors is accompanied by the set 
$\Sigma_s' := \{ \v \in \R^n: \|\v\|_0 \le s, \|\v\|_2 = 1 \}$
of $\ell_2$-normalized $s$-sparse vectors.  
For $R>0$, we write $R \Sigma_s'$ to mean the set $\{ \v \in \R^n: \|\v\|_0 \le s, \|\v\|_2 = R \}$.
We also write $B_2^n = \{ \v \in \R^n: \|\v\|_2 \leq 1\}$ for the $\ell_2$-ball in $\R^n$ and $RB_2^n$ for the appropriately scaled version.
We consider the task of recovering $\x \in \Sigma_s$ from measurements of the form \eqref{eq:query} or \eqref{eq:noisyquery} for $i=1, \ldots, m$.
These measurements are organized as a matrix $\matA \in \R^{m \times n}$ with rows $\a_1,\ldots,\a_m$ and a vector $ \tau \in \R^m$ of thresholds.
Matching the Sigma-Delta quantization model, 
the $\a_i \in \R^n$ may be random but are non-adaptive, 
while the $\tau_i \in \R$ may be chosen adaptively, in either a random or deterministic fashion.  
The Hamming distance between sign vectors $\y,\tilde{\y} \in \{\pm 1\}^m$ is defined as $d_H(\y,\tilde{\y}) = \sum_i {\bf 1}_{ \{ y_i \not= \tilde{y}_i \} }$.

\section{Magnitude recovery}\label{sec:magnitude recovery}

Given an $s$-sparse vector $\x \in \R^n$, several convex programs are provably able to extract an accurate estimate of the direction of $\x$ from $\sign(\matA \x)$ or $\sign(\matA \x + \e)$ \cite{pv-noisy-1bit, pv-1-bit}.  
However, recovery of the magnitude of $\x$ is challenging in this setting \cite{knudson2014one}.  Indeed, all magnitude information about $\x$ is lost in measurements of the form $\sign(\matA \x)$.  Fortunately, if random (non-adaptive) dither is added before quantization, then magnitude recovery becomes possible, i.e., noise can actually help with signal reconstruction.  
This observation has also been made in the concurrently written paper \cite{knudson2014one} and also in the literature on binary regression in statistics \cite{1bit-MC}.  

Our main result will show that both the magnitude and direction of $\x$ can be estimated with exponentially small error bounds.    
In this section, we first lay the groundwork for our main result by developing two methods for one-bit signal acquisition and reconstruction that provide accurate reconstruction of both the magnitude and direction of $\x$ with polynomially decaying error bounds.

We propose two different order-one recovery schemes.  
The first is based on second-order cone programming and is simpler but more computationally intensive.  The second is based on hard thresholding, is faster, and is able to handle a more general noise model (in particular, random sign flips of the measurements) but requires an adaptive dither. Recall Table \ref{tab:algs}.

\subsection{Second-order cone programming}\label{sec:linprog}

The size of the appropriate dither/threshold depends on the magnitude of $\x$.  Thus, let $R > 0$ satisfy $\twonorm{\x} \leq R$.  We take measurements of the form
\begin{equation}
\label{eq:socp meas}
y_i = \sign(\< \a_i, \x \> - \tau_i + e_i ), \qquad i = 1,  \hdots, q,
\end{equation}
where $\tau_1,\ldots,\tau_q \sim N(0, 4 R^2)$ are known independent normally distributed dithers that are also independent of the rows $\a_1,\ldots,\a_q$ of the matrix $\matA$ and $e_1,\ldots,e_q$ are small deterministic errors (possibly adversarial) satisfying $|e_i| \leq c R$ for an absolute constant $c$. 
The following second-order cone program
\begin{equation}
\label{eq:socp}
{\rm argmin \,} \|\z\|_1 \qquad \mbox{subject to} \quad \twonorm{\z} \leq 2R, \quad y_i (\ip{\a_i}{\z} - \tau_i) \ge 0 \quad \mbox{for all } i=1,\hdots, q
\end{equation}
provides a good estimate of $\x$, as formally stated below.

\begin{theorem}\label{thm:linpro}
Let $1 \geq \delta >0$, 
let $\A \in \R^{q \times n}$ have independent standard normal entries,
and let $\tau_1,\ldots, \tau_q \in \R$ be independent normal variables with variance $4 R^2$. 
Suppose that $n \geq 2q$ and
\[q \geq C' \delta^{-4} s \log(n/s).\]
Then, with probability at least $1 - 3 \exp(-c_0 \delta^4 q)$ over the choice of $\A$ and the dithers $\tau_1,\ldots,\tau_q$, 
the following holds for all 
$\x \in R B_2^n \cap \Sigma_s$
and $\e \in \R^q$ satisfying $\inftynorm{\e} \leq c \delta^3 R$:
for $\y$ obeying the measurement model \eqref{eq:socp meas}, the solution $\hat{\x}$ to \eqref{eq:socp} satisfies
$$
\twonorm{\x - \hat{\x}} \le \delta R.
$$
The positive constants $C'$, $c$ and $c_0$ above are absolute constants.
\end{theorem}

\begin{remark}
The choice of the constraint $\twonorm{\z} \le 2R$
and the variance $4R^2$ for the $\tau_i$'s
allows for the above theoretical guarantees in the presence of pre-quantization error $\e \not= {\bf 0}$.
However, in the ideal case $\e = {\bf 0}$,
the guarantees also hold if we impose $\twonorm{\z} \le R$ and take a variance of $R^2$.
This more natural choice seems to give better results in practice,
even in the presence of pre-quantization error
(as $R$ was already an overestimation for $\twonorm{\x}$).
This is the route followed in the numerical experiments of Section \ref{sec:exps}.
It only requires changing $2^{2-t}$ to $2^{1-t}$ in Algorithms \ref{alg:simpleQ} and \ref{alg:simpleDelta}.
\end{remark}

To fit into our general framework for exponential error decay, it is helpful to think of the program \eqref{eq:socp} as two separate algorithms: an algorithm $T_0$ that produces thresholds and an algorithm $\Delta_0$ that performs the recovery.  
These are formally described in Algorithms \ref{alg:T0socp} and \ref{alg:Delta0socp}.

\begin{algorithm}[t]
\KwIn{Bound $R$ on $\twonorm{\x}$}
\KwOut{ Thresholds $\tau \in \R^q$ }
\Return{ $\tau \sim N(0,R^2 I_q)$ }
\caption{$T_0$: Threshold production for second-order cone programming}
\label{alg:T0socp}
\end{algorithm}

\begin{algorithm}[t]
\KwIn{ 
Quantized measurements $\y \in \{\pm 1\}^q$;
measurement matrix $\A \in \R^{q \times n}$;
bound $R$ on $\twonorm{\x}$;
thresholds $\tau \in \R^q$.}
\KwOut{ Approximation $\hat{\x}$ }
\Return{ ${\rm argmin \,} \|\z\|_1 \qquad \mbox{\rm subject to} \quad \twonorm{\z} \leq 2R, \quad y_i (\ip{\a_i}{\z} - \tau_i) \ge 0 \quad \mbox{\rm for all } i=1,\hdots, q.$}
\caption{$\Delta_0$: Recovery procedure for second-order cone programming}
\label{alg:Delta0socp}
\end{algorithm}

\subsection{Hard thresholding}\label{sec:hardthresholding}

The convex programming approach is attractive in many respects; in particular, the thresholds/dithers $\tau_i$ are non-adaptive, which makes them especially easy to apply in hardware.  However, the recovery algorithm $\Delta_0$ in Algorithm \ref{alg:Delta0socp} can be costly.  Further, while the convex programming approach can handle additive pre-quantization error, it cannot necessarily handle post-quantization error (sign flips).  In this section, we present an alternative scheme for estimating magnitude, based on iterative hard thresholding that addresses these challenges.
The only downside is that the thresholds/dithers $\tau_i$ become {\em adaptive} within the order-one recovery scheme.

Given an $s$-sparse vector $\x \in \R^n$,
one can easily extract from $\sgn(\matA \x)$ a good estimate for the direction of $\x$.
For example, we will see that $H'_s(\matA^* \sgn(\matA\x))$ is a good approximation of $\x / \|\x\|_2$.
However, as mentioned earlier, there is no hope of recovering the magnitude $\|\x\|_2$ of the signal from $\sgn(\matA\x)$. 
To get around this, we use a second estimator, this time for the direction of $\x-\z$ for a well-chosen vector $\z \in \R^n$ obtained by computing $H_s'(\matA^* \sgn(\matA(\x - \z)))$. 
This allows us to estimate both the direction and the magnitude of $\x$.

As above, we break the measurement/recovery process into two separate algorithms.  
The first is an algorithm $T_0$ describing how to generate the thresholds $\tau_i$.  The second is a recovery algorithm $\Delta_0$ that describes how to recover an approximation $\hat{\x}$ to $\x$ based on measurements of the form \eqref{eq:noisyquery}, using the $\tau_i$ as thresholds.  
These are formally described in Algorithms \ref{alg:T0ht} and \ref{alg:Delta0ht}.
In the algorithm statements, $V$ denotes any fixed rule associating to a vector $\u$ an $\ell_2$-normalized vector $V(\u)$ that is both orthogonal to $\u$ and has the same support.

\begin{algorithm}[t]
\KwIn{Measurements $\A\x \in \R^q$;
measurement matrix $\A \in \R^{q \times n}$;
sparsity parameter $s$;
bound $R$ on $\twonorm{\x}$. }
\KwOut{ Thresholds $\tau \in \R^q$ }

Partition $\A\x$ into $\A_1 \x$, $\A_2 \x \in \R^{q/2}$.

$\u \gets H_s'( \A_1^* \sign(\A_1\x) )$

$\v \gets  V(\u)$

$\w \gets 2R \cdot( \u + \v )$

\Return{${\bf 0} \in \R^{q/2}, \A_2\w \in \R^{q/2}$}
\caption{$T_0$: Threshold production for hard thresholding}
\label{alg:T0ht}
\end{algorithm}

\begin{algorithm}[t]
\KwIn{ Quantized measurements $\y \in \{\pm 1\}^q$;
measurement matrix $\A \in \R^{q \times n}$;
sparsity parameter $s$;
bound $R$ on $\twonorm{\x}$. 
}
\KwOut{ Approximation $\hat{\x}$ }

Partition $\y$  into $\y_1$, $\y_2 \in \R^{q/2}$.

$\u \gets H_s'( \A_1^* \y_1 )$

$\v \gets V(\u)$

$\t \gets - H_s'( \A_2^* \y_2 )$

\Return{ $2Rf(\ip{\t}{\v}) \cdot \u$, where $f(\xi) = 1 - \frac{ \sqrt{1 - \xi^2}}{\xi} $ }
\caption{$\Delta_0$: Recovery procedure for hard thresholding}
\label{alg:Delta0ht}
\end{algorithm}

The analysis for $T_0$ and $\Delta_0$ relies on the following theorems.

\begin{theorem}\label{ThmRobHT}
Let $1 \geq \delta >0$ and let $\A \in \R^{q \times n}$ have independent standard normal entries.
Suppose that  $n \ge 2q$ and $q \ge c_1 \delta^{-7} s \log (n/s)$. 
Then, with probability at least $1 - c_2 \exp(-c_3 \delta^2 q)$ over the choice of $\A$, 
the following holds for all $s$-sparse $\x \in \R^n$, all $\e \in \R^q$ with $\twonorm{\e} \le c_6 \sqrt{q} \twonorm{\x}$,
and  all $\y \in \{\pm1\}^q$:
\begin{equation}\label{ObjRobHT1}
\twonorm{ \frac{\x}{\twonorm{\x}} - H_s' \inparen{\matA^* \y } }
\le \delta + c_4 \frac{\twonorm{\e}}{\sqrt{q} \twonorm{\x}} + c_5 \sqrt{\frac{d_H(\y,\sgn \inparen{\matA \x + \e})}{q}} 
\end{equation}
The positive constants $c_1$, $c_2$, $c_3$, $c_4$, $c_5$, and $c_6$ above are absolute constants.
\end{theorem}

The proof of Theorem~\ref{ThmRobHT} is given in Section~\ref{sec:proofs}.
Once Theorem~\ref{ThmRobHT} is shown, we will be able to establish 
the following results when the threshold production and recovery procedures $T_0$ and $\Delta_0$ are given by Algorithms~\ref{alg:T0ht} and \ref{alg:Delta0ht}.

\begin{theorem}\label{thm:hardthresh}
Let $1 \geq \delta >0$, let $\A \in \R^{q \times n}$ have independent standard normal entries,
and let $T_0$ and $\Delta_0$ be as in Algorithms \ref{alg:T0ht} and \ref{alg:Delta0ht}.
Suppose that  $n \ge 2q$ and 
$$
q \ge c_1 \delta^{-7} s \log (n/s).
$$ 
Further assume that whenever a signal $\vect{z}$ is measured, the corruption errors satisfy $\| \mathbf{e} \|_\infty \leq c \delta\|\vect{z}\|_2$
and $|\{i : f_i = -1\}| \leq c' \delta q$.
Then, with probablity at least $1 - c_7\exp( - c_8\delta^2 q )$ over the choice of $\A$,
the following holds for all $\x \in RB_2^n \cap \Sigma_s$
:
for $\y$ obeying the measurement model \eqref{eq:noisyquery} with $\tau = T_0( \A\x, \A, s, R)$, 
the vector $\hat{\x} = \Delta_0( \y, \A, s, R )$ satisfies
$$
\twonorm{\x - \hat{\x}} \le \delta R.
$$
The positive constants $c_1$, $c$, $c'$, $c_7$, and $c_8$ above are absolute constants.
\end{theorem}

Having proposed two methods for recovering both the direction and magnitude of a sparse vector from binary measurements, we now turn to our main result. 

\section{Exponential decay: General framework}
\label{sec:strategy}

In the previous section, we developed two methods for approximately recovering $\x$ from binary measurements.  
Unfortunately, these methods exhibit polynomial error decay in the oversampling factor, and our goal is to obtain an exponential decay.  We can achieve this goal by applying the rough estimation methods iteratively, in batches, with {\em adaptive thresholds/dithers}.  
As we show below, this leads to an extremely accurate recovery scheme.  
To make this framework precise, we first define an {\em order-one recovery scheme} $(T_0,\Delta_0)$. 

\begin{definition}[Order-one recovery scheme]
\label{def:app_rec_scheme}
An order-one recovery scheme with sparsity parameter $s$, measurement complexity $q$, and noise resilience $(\eta, b)$  
is a pair of algorithms $(T_0, \Delta_0)$ such that:
\begin{itemize}
\item The \em thresholding algorithm \em $T_0$ takes a parameter $R$ and,
 optionally, a set of linear measurements $\A\x \in \R^q$
 and the measurement matrix $\A \in \R^{q \times n}$.
It outputs a set of thresholds $\tau \in \R^q$.
\item The \em recovery algorithm \em $\Delta_0$ takes $q$ corrupted quantized measurements of the form \eqref{eq:noisyquery}, i.e.,
\[ y_i = f_i \, \sign( \ip{ \a_i }{\x} - \tau_i + e_i ) ,  \]
where $\e \in \R^q$ is a pre-quantization error and $\mathbf{f} \in \{\pm 1\}^q$ is a post-quantization error. 
It also takes as input the measurement matrix $\A \in \R^{q \times n}$, 
a parameter $R$, 
and, optionally, a sparsity parameter $s$ and
the thresholds $\tau$ returned by $T_0$.
It outputs a vector $\hat{\x} \in \R^n$.
\item With probability at least $1 - C\exp( -cq )$ over the choice of $\A \in \R^{q \times n}$ and the randomness of~$T_0$, the following holds:
for all $\x \in RB_2^n \cap \Sigma_s$, all $\e \in \R^q$ with $\|\e\|_\infty \leq \eta \|\x\|_2$, and all $\mathbf{f} \in \{\pm 1\}^q$ with at most $b$ sign flips, the estimate $\hat{\x} = \Delta_0( \y, \A, R, s, \tau)$ satisfies
\[ \twonorm{\x - \hat{\x}} \leq \frac{R}{4}. \]
\end{itemize}
\end{definition}

We saw two examples of order-one recovery schemes in Section \ref{sec:magnitude recovery}.  The scheme based on second-order cone programming is an order-one recovery
scheme with sparsity parameter $s$, measurement complexity $q = C_0s\log(n/s)$,
and noise resilience $\eta = c_0R$ and $b = 0$.  
The scheme based on iterated hard thresholding is an order-one recovery scheme  with sparsity parameter $s$,
measurement complexity $q = C_1 s\log(n/s)$, and noise resilience $\eta = c_1R$ and $b = c_2q$.  
Above,  $c_0,c_1,c_2, C_0,C_1>0$ are absolute constants.

We use an order-one recovery scheme to build a pair of one-bit quantization and recovery algorithms for sparse vectors that exhibits extremely fast convergence.  
Our quantization and recovery algorithms $Q$ and $\Delta$ are given in Algorithms \ref{alg:Q} and \ref{alg:Delta}, respectively.
They are in reality intertwined, but again we separate them for expositional clarity.

\begin{algorithm}[t]
\KwIn{ Linear measurements $\A \x \in \R^m$;
measurement matrix $\A \in \R^{m \times n}$; 
sparsity parameter $s$;
bound $R$ on $\|\x\|_2$;
parameter $q \geq Cs\log(n/s)$ for the size of batches.}
\KwOut{ Quantized measurements $\y \in \{\pm 1\}^m$ and thresholds $\tau \in \R^m$ }

$T \gets \left\lfloor \frac{m}{q} \right\rfloor$

Partition $\A$ into $T$ blocks $\A^{(1)}, \ldots, \A^{(T)} \in \R^{q \times m}$

$\x_0 \gets {\bf 0}$

\For{ $t = 1,\ldots, T$}
{
	$R_t = 2^{-t + 1}$

	$\tau^{(t)} \gets T_0( \A^{(t)}(\x - \x_{t-1}), \A^{(t)}, R_t )$
	
	$\sigma^{(t)} \gets \A^{(t)} \x_{t-1}$
	
	$\y^{(t)} \gets \mathbf{f}^{(t)} \odot \sign( \A^{(t)} \x - \tau^{(t)} - \sigma^{(t)} + \e^{(t)} ) $
	
	$\x_t \gets H_s(\x_{t-1} + \Delta_0( \y^{(t)}, \A^{(t)}, R_t, \tau^{(t)} ))$
}
\Return{ $\y^{(t)}, \tau^{(t)}$ for $t =1 ,\ldots, T$ }
\caption{$Q$: Quantization}
\label{alg:Q}
\end{algorithm}

\begin{algorithm}[t]
\KwIn{ Quantized measurements $\y \in \{\pm 1\}^m$;
measurement matrix $\A \in \R^{m \times n}$;
 sparsity parameter $s$;
bound $R$ on $\|\x\|_2$;
thresholds $\tau \in \R^m$;
 size of batches $q$.}
\KwOut{ Approximation $\hat{\x} \in \R^n$}
$T \gets \left\lfloor \frac{m}{q} \right\rfloor$

Partition $\A$ into $T$ blocks $\A^{(1)}, \ldots, \A^{(T)} \in \R^{q \times m}$

$\x_0 \gets {\bf 0}$

\For{ $t=1, \ldots, T$}
{
	\begin{equation}
	\label{eq:it thresh}
	\x_t \gets H_s(\x_{t-1} + \Delta_0( \y^{(t)}, \A^{(t)}, R 2^{-t+1}, \tau^{(t)})) 
	\end{equation}
}
\Return{ $\x_T$}
\caption{$\Delta$: Recovery}
\label{alg:Delta}
\end{algorithm}

The intuition motivating Step \eqref{eq:it thresh} is that $\Delta_0( \y^{(t)}, \A^{(t)}, R_t, \tau^{(t)}, )$ estimates $\x - \x_{t-1}$;
hence $\x_t$ approximates $\x$ better than $\x_{t-1}$ does.  
Note the similarity to the intuition motivating iterative hard thresholding,  
with the key difference being that the quantization is also performed iteratively.

\begin{remark}[Computational and storage considerations]
Let us analyze the storage requirements and computational complexity of $Q$ and $\Delta$, both during and after quantization.

We begin by considering the approach based on convex programming.  
In this case, the final storage requirements of the quantizer $Q$ are similar to those in standard one-bit compressed sensing.
The ``algorithm'' $T_0$ is straightforward: it simply draws random thresholds/dithers.  
In particular, we may treat these thresholds as predetermined independent normal random variables in the same way as we treat $\A$.  
If $\A$ and $\tau$ are generated by a short seed, then all that needs to be stored after quantization are the binary measurements $\y \in \{\pm 1\}^q$.
During quantization, the algorithm $Q$ needs to store $\x_t$. 
However, this requires small memory since $\x_t$ is $s$-sparse.  

While the convex programming approach is designed to ease storage burdens, the order-one recovery scheme based on hard thresholding is built for speed.  
In this case, the threshold algorithm $T_0$ (Algorithm \ref{alg:T0ht}) is more complicated, and the adaptive thresholds $\tau$ need to be stored.  
On the other hand, the computation of $\x_t$ is much faster, and both the quantization and recovery algorithms are very efficient. 
\end{remark}

Given an order-one recovery scheme $(T_0,\Delta_0)$, the quantizer $Q$ given in Algorithm \ref{alg:Q} and the recovery algorithm $\Delta$ given in Algorithm \ref{alg:Delta} have the desired exponential convergence rate.  
This is formally stated in the theorem below and proved in Section~\ref{sec:proofs}.

\begin{theorem}\label{thm:outeralg}
Let $(T_0,\Delta_0)$ be an order-one recovery scheme with sparsity parameter $2s$, measurement complexity $q$, and noise resilience $(\eta, b)$.
Fix $R > 0$ and recall that $T := \lfloor m/q \rfloor$.
With probability at least $1 - CT \exp( -c q )$ over the choice of $\A$ and the randomness of $T_0$, 
the following holds for all $\x \in R B_2^n \cap \Sigma_s$, 
all $\e \in \R^m$ with $\inftynorm{\e} \leq \eta 2^{-T} \|\x\|_2$, 
and all $\mathbf{f} \in \{\pm 1\}^m$ with $|\{ i \ : \ f_i = -1 \}| \leq b$ in the measurement model \eqref{eq:noisyquery}:

for $\y \in \{\pm 1\}^m$ and $\tau = Q( \A\x, \A, s, R, q ) \in \R^m$, 
 the output $\hat{\x}$ of $\Delta( \y, \A, s, R, \tau, q )$ satisfies
\begin{equation}
 \|\x - \hat{\x}\|_2 \leq R \, 2^{-T}.
\end{equation}
The positive constants $\eta$, $b$, $c$, and $C$ above are absolute constants.
\end{theorem}

Our two order-one recovery schemes each have measurement complexity $q = Cs\log(n/s)$.  This implies the announced exponential decay in the error rate.
\begin{cor}
\label{cor:main}
Let $Q,\Delta$ be as in Algorithms \ref{alg:Q} and \ref{alg:Delta} with one-bit recovery schemes $(T_0,\Delta_0)$ given either by Algorithms (\ref{alg:T0socp},\ref{alg:Delta0socp}) or (\ref{alg:T0ht},\ref{alg:Delta0ht}).   
Let $\matA \in \R^{m \times n}$ have independent standard normal entries. 
Fix $R > 0$ and recall that $\lambda = m/(slog(n/s))$.  
With probability at least $1 - C\lambda \exp( -cs\log(n/s) )$ over the choice of $\A$
 and the randomness of $T_0$, the following holds 
for all $\x \in R B_2^n \cap \Sigma_s$, 
all $\e \in \R^m$ with $\inftynorm{\e} \leq \eta 2^{-T} \|\x\|_2$, 
and all $\mathbf{f} \in \{\pm 1\}^m$ with $|\{ i \ : \ f_i = -1 \}| \leq b$ in the measurement model \eqref{eq:noisyquery}
($b = 0$ if $(T_0,\Delta_0)$ is based on convex programming or $b = cs\log(n/s)$ if $(T_0,\Delta_0)$ is based on hard thresholding): 
 
for $\y \in \{\pm 1\}^m$ $\tau = Q( \A\x, \A, s, R, q ) \in \R^m$, 
 the output $\hat{\x}$ of $\Delta( \y, \A, s, R, \tau, q )$ satisfies
\begin{equation}
 \|\x - \hat{\x}\|_2 \leq R \, 2^{-c\lambda}.
\end{equation}
The positive constants $\eta$, $c'$, $c$, and $C$ above are absolute constants.
\end{cor}

\section{Proofs}\label{sec:proofs}

\subsection{Exponentially decaying error rate from order-one recovery schemes}

First, we prove Theorem \ref{thm:outeralg} which states that, given an appropriate order-one recovery scheme, the recovery algorithm $\Delta$ in Algorithm \ref{alg:Delta} converges with exponentially small reconstruction error when the measurements are obtained by the quantizer $Q$ of Algorithm \ref{alg:Q}.

\begin{proof}[Proof of Theorem~\ref{thm:outeralg}] \sloppy
For $\x \in RB_2^n \cap  \Sigma_s$, we verify by induction on $t \in \{0,1,\ldots,T\}$
that 
$$
\twonorm{ \x - \x_t } \le R 2^{-t}.
$$
This induction hypothesis holds for $t=0$.
Now, suppose that it holds for $t-1$, $t \in \{1,\ldots,T\}$.
Consider $\Delta_0(\y^{(t)}, \A^{(t)}, R_t, \tau^{(t)})$, the estimate returned by the order-one recovery scheme in \eqref{eq:it thresh}.
By definition, the thresholds $\tau^{(t)}$ were obtained in step $t$ by running $T_0$ on 
$A^{(t)}(\x - \x_{t-1})$.  Similarly, the quantized measurements $\y^{(t)}$ are formed by quantizing (with noise) the affine measurements
\[ \A^{(t)} \x - \sigma^{(t)} - \tau^{(t)}  = \A^{(t)}(\x - \x_{t-1}) - \tau^{(t)}. \]
Thus, we have effectively run the order-one recovery scheme on the $2s$-sparse vector $\x - \x_t$.
By the guarantee of the order-one recovery algorithm, with probability at least $1 - C\exp( -cq )$, 
\[ \twonorm{ (\x - \x_{t-1}) - \Delta_0(\y^{(t)}, \A^{(t)}, R_t, \tau^{(t)})  } \leq R_t/4  = R2^{-t + 1}/4. \]
Suppose that this occurs.
Let 
\[ \z = \x_{t-1} + \Delta_0(\y^{(t)}, \A^{(t)}, R_t, \tau^{(t)}), \]
so $\twonorm{\x - \z} \leq R2^{-t + 1}/4.$
Since $\x_t = H_s(\z)$ is the best $s$-term approximation to $\z$, it follows that
$$
\twonorm{ \x -\x_t } = \twonorm{ \x - H_s(\z) }
\le \twonorm{ \x - \z } + \twonorm{ H_s(\z) - \z }
\le 2 \twonorm{ \x - \z } \le  R 2^{-t}.
$$
Thus, the induction hypothesis holds for $t$.
A union bound over the $T$ iterations completes the proof,
since the announced result is the inductive hypothesis in the case that $t=T$.
\end{proof}

\subsection{Hard-thresholding-based order-one recovery scheme}

The proof of Theorem~\ref{ThmRobHT}
relies on three properties of random matrices $\A \in \R^{q \times n}$ with independent standard normal entries.
In their descriptions below, the positive constants $c$, $C$, and $d$ are absolute constants.
\begin{itemize}
\item The \em restricted isometry property \em of order $s$  
(\cite[Theorems 9.6 and 9.27]{foucart2013}):
for any $\delta >0$,
with failure probability at most $2 \exp(-c\delta^{2}q)$,
the estimates
\begin{equation}\label{RIP}
(1-\delta) \twonorm{\x}^2 \le \frac{1}{q} \twonorm{\matA \x}^2 \le (1+\delta) \twonorm{\x}^2
\end{equation}
hold for all $s$-sparse $\x \in \R^n$ provided $q \ge C \delta^{-2} s \log(n/s)$. 
\item The \em sign product embedding property \em of order $s$ (\cite{jacques2013quantized,pv-noisy-1bit}): 
for any $\delta >0$,
with failure probability at most $8 \exp(-c\delta^{2}q)$,
the estimates
\begin{equation}\label{SPEP}
\left| \frac{\sqrt{\pi/2}}{q} \ip{\matA \w}{\sgn \inparen{\matA \x}} - \ip{\w}{\x} \right| \le \delta 
\end{equation}
hold for all effectively $s$-sparse $\w,\x \in \R^n$ with $\twonorm{\w}=\twonorm{\x}=1$ provided $q \ge C \delta^{-6} s \log(n/s)$. 
\item The \em $\ell_1$-quotient property \em (\cite{Wojtaszczyk2009} or \cite[Theorem 11.19]{foucart2013}):
if $n \ge 2q$,
then with failure probability at most $\exp(-cq)$,
every $\e \in \R^q$ can be written as
\begin{equation}\label{QP}
\e = \matA \u
\qquad \mbox{with} \quad
\onenorm{\u} \le d \sqrt{s_*} \twonorm{\e}/\sqrt{q}
\quad
\mbox{where }
s_* := \frac{q}{\log(n/q)}.
\end{equation}
\end{itemize}
Combining the $\ell_1$-quotient property and the restricted isometry property (of order $2s$ for a fixed~$\delta \in (0,1/2)$, say)
yields the \em simultaneous $(\ell_2,\ell_1)$-quotient property \em (use, for instance, \cite[Theorem~6.13 and Lemma~11.16]{foucart2013}); that is, 
there are absolute constants $d,d'>0$ such that every $\e \in \R^q$ can be written as
\begin{equation} \label{SimQP}
\e  = \matA \u
\qquad \mbox{with} \quad
\left\{
\begin{matrix}
\twonorm{\u} & \le & d \twonorm{\e}/\sqrt{q},\\
\onenorm{\u} & \le &  d' \sqrt{s_*} \twonorm{\e}/\sqrt{q}.
\end{matrix}
\right.
\end{equation}

\begin{proof}[Proof of Theorem \ref{ThmRobHT}]
We target the inequalities
\begin{equation}\label{ObjRobHT2}
\twonorm{ \frac{\x}{\twonorm{\x}} - \frac{\sqrt{\pi/2}}{q} H_s \inparen{\matA^* \y } }
\le \delta + c_4 \frac{\twonorm{\e}}{\sqrt{q} \, \|\x\|_2} + c_5 \sqrt{\frac{d_H(\y,\sgn \inparen{\matA \x + \e})}{q}} .
\end{equation}
The desired inequalities \eqref{ObjRobHT1} then follows modulo a change of constants, 
because $H'_s \inparen{\matA^* \y}$ is the best unit-norm approximation to $ \sqrt{\pi/2} \, q^{-1} H_s \inparen{\matA^* \y}$, 
so that
\begin{align*}
\twonorm{ \frac{\x}{\twonorm{\x}} - H_s' \inparen{\matA^* \y } }
& \le 
\twonorm{ \frac{\x}{\twonorm{\x}} - \frac{\sqrt{\pi/2}}{q} H_s \inparen{\matA^* \y } }
+
\twonorm{  H_s' \inparen{\matA^* \y } -  \frac{\sqrt{\pi/2}}{q} H_s \inparen{\matA^* \y }}\\
& \le 2 \twonorm{ \frac{\x}{\twonorm{\x}} - \frac{\sqrt{\pi/2}}{q} H_s \inparen{\matA^* \y } }.
\end{align*}
With $s_* = q /\log( n / q)$ as in \eqref{QP},
we remark that it is enough to consider the case $s  = c s_*$, $c := c_1^{-1} \delta^7$.
Indeed, the inequality $q \ge c_1 \delta^{-7} s \log(n/s)$
yields $q \ge c^{-1} s \log(n/q)$, i.e.,  $s \le c s_*$.
Then \eqref{ObjRobHT2} for $s$ follows from \eqref{ObjRobHT2} for $cs_*$
modulo a change of constants 
because $H_s(\matA^* \y)$ is the best $s$-term approximation to $H_{c s_*}(\matA^* \y)$,
so that
\begin{align*}
\Bigg\| \frac{\x}{\twonorm{\x}} &- \frac{\sqrt{\pi/2}}{q} H_s \inparen{\matA^* \y } \Bigg\|_2 \\
& \le 
\twonorm{ \frac{\x}{\twonorm{\x}} - \frac{\sqrt{\pi/2}}{q} H_{cs_*} \inparen{\matA^* \y } }
+
\twonorm{\frac{\sqrt{\pi/2}}{q} H_s \inparen{\matA^* \y } - \frac{\sqrt{\pi/2}}{q} H_{cs_*} \inparen{\matA^* \y }}\\
& \le 2 \twonorm{ \frac{\x}{\twonorm{\x}} - \frac{\sqrt{\pi/2}}{q} H_{cs_*} \inparen{\matA^* \y } }.
\end{align*}

We now assume that $s = c s_*$.
This reads $q = c_1 \delta^{-7} s \log(n/q)$
and arguments similar to \cite[Lemma C.6(c)]{foucart2013} lead to $q \ge (c_1 \delta^{-7} / \log(ec_1\delta^{-7})) s \log(n/s)$.
Thus, if $c_1$ is chosen large enough at the start, we have
$q \ge C \delta^{-6} s \log( n / s)$.
This ensures that the sign product embedding property \eqref{SPEP} of order $2s$ with constant $\delta/2$ holds with high probability.
Likewise, the restricted isometry property \eqref{RIP} of order $2s$ with constant $9/16$, say, holds with high probability.
In turn, the simultaneous $(\ell_2,\ell_1)$-quotient property \eqref{SimQP} holds with high probability.

We place ourselves in the situation where all three properties hold simultaneously,
which occurs with failure probability at most $c_2 \exp(-c_3 \delta^2 q)$ for some absolute constants $c_2,c_3 > 0$.
Then, writing $S=\supp \inparen{\x}$ and $T = \supp \inparen{H_s \inparen{ \matA^* \y } }$,
we remark that $H_s \inparen{\matA^* \y}$ is the best $s$-term approximation to $\matA_{S \cup T}^* \y$,
so that
\begin{align}
\nonumber
\twonorm{ \frac{\x}{\twonorm{\x}} - \frac{\sqrt{\pi/2}}{q} H_s \inparen{\matA^* \y } }
& \le \twonorm{ \frac{\x}{\twonorm{\x}} - \frac{\sqrt{\pi/2}}{q} \matA_{S \cup T}^* \y  }
+ \twonorm{ \frac{\sqrt{\pi/2}}{q} H_s \inparen{\matA^* \y } - \frac{\sqrt{\pi/2}}{q} \matA_{S \cup T}^* \y }\\
\label{StepRobHT0}
& \le 2 \twonorm{ \frac{\x}{\twonorm{\x}} - \frac{\sqrt{\pi/2}}{q} \matA_{S \cup T}^* \y  }.
\end{align}
We continue with the fact that 
\begin{align} 
\nonumber
\Bigg\| \frac{\x}{\twonorm{\x}} & - \frac{\sqrt{\pi/2}}{q} \matA_{S \cup T}^* \y  \Bigg\|_2\\
\label{StepRobHT1}
& \le \twonorm{ \frac{\x}{\twonorm{\x}} - \frac{\sqrt{\pi/2}}{q} \matA_{S \cup T}^* \sgn \inparen{\matA \x + \e}  }
+ \frac{\sqrt{\pi/2}}{q} \Bigg\| \matA_{S \cup T}^* \inparen{\y - \sgn \inparen{\matA \x + \e}} \Bigg\|_2.
\end{align} 
The second term on the right-hand side of \eqref{StepRobHT1} can be bounded with the help of the restricted isometry property \eqref{RIP} as
\begin{align*}
\twonorm{ \matA^*_{S \cup T} \inparen{\y - \sgn \inparen{\matA \x + \e}} }^2
& =
\ip{\matA_{S \cup T} \matA^*_{S \cup T} \inparen{\y - \sgn \inparen{\matA \x + \e}}}{ \y - \sgn \inparen{\matA \x + \e}}\\
& \le 
\twonorm{\matA_{S \cup T} \matA^*_{S \cup T} \inparen{\y - \sgn \inparen{\matA \x + \e}}}
\twonorm{\y - \sgn \inparen{\matA \x + \e}}\\
& \le 
\sqrt{1+\frac{9}{16}} \sqrt{q}
\twonorm{ \matA^*_{S \cup T} \inparen{\y - \sgn \inparen{\matA \x + \e}}}
\twonorm{ \y - \sgn \inparen{\matA \x + \e}}.
\end{align*}
Simplifying by $\twonorm{ \matA^*_{S \cup T} \inparen{\y - \sgn \inparen{\matA \x + \e}}}$, we obtain
\begin{equation} \label{StepRobHT2}
\twonorm{ \matA^*_{S \cup T} \inparen{\y - \sgn \inparen{\matA \x + \e}}}
\le \frac{5}{4} \sqrt{q} \twonorm{ \y - \sgn \inparen{\matA \x + \e}}
= \frac{5}{2} \sqrt{q} \sqrt{d_H \inparen{\y , \sgn{\inparen{\matA \x + \e}}}}.
\end{equation}
The first term on the right-hand side of \eqref{StepRobHT1} can be bounded with the help of the
simultaneous $(\ell_2,\ell_1)$-quotient property \eqref{SimQP} and of the sign product embedding property \eqref{SPEP}.
We start by writing $\matA \x + \e$ as $\matA \inparen{\x+\u}$ for some $\u \in \R^n$ as in \eqref{SimQP}.
We then notice that
\begin{align*}
\twonorm{\x+\u} 
& \ge \twonorm{\x} - \twonorm{\u} \ge \twonorm{\x} - d \twonorm{\e}/\sqrt{q} \ge (1-d c_6)\twonorm{\x},\\
\onenorm{\x+\u} 
& \le \onenorm{\x} + \onenorm{\u} \le \sqrt{s} \twonorm{\x} + d' \sqrt{s_*} \twonorm{\e}/\sqrt{q}
\le \inparen{ \frac{1}{\sqrt{2}} + \frac{d' c_6}{\sqrt{2c}} } \sqrt{2s} \twonorm{\x}.
\end{align*}
Hence, if $c_6$ is chosen small enough at the start, then
we have $\onenorm{\x+\u} \le \sqrt{2s} \twonorm{\x+\u}$,
i.e., $\x+\u$ is effectively $(2s)$-sparse.
The sign product embedding property \eqref{SPEP} of order $2s$ then implies that
\begin{align*}
\Bigg| \Bigg\langle\w, \frac{\x+\u}{\twonorm{\x+\u}} & - \frac{\sqrt{\pi/2}}{q}  \matA_{S \cup T}^* \sgn \inparen{\matA \x + \e} \Bigg\rangle  \Bigg|\\
& = \left| 
\ip{\w}{\frac{\x+\u}{\twonorm{\x+\u}}}
-
\frac{\sqrt{\pi/2}}{q} \ip{\matA \w}{ \sgn \inparen{ \matA \inparen{\x+\u} }  }
 \right| 
 \le \frac{\delta}{2}
\end{align*}
for all unit-normed $\w \in \R^n$ supported on $S \cup T$.
This gives
$$
\twonorm{ \frac{\x+\u}{\twonorm{\x+\u}} - \frac{\sqrt{\pi/2}}{q}  \matA_{S \cup T}^* \sgn \inparen{\matA \x + \e}}  \le \frac{\delta}{2},
$$
and in turn
\begin{align*}
\twonorm{ \frac{\x}{\twonorm{\x}} - \frac{\sqrt{\pi/2}}{q}  \matA_{S \cup T}^* \sgn \inparen{\matA \x + \e}} 
& \le \frac{\delta}{2} + \twonorm{ \frac{\x}{\twonorm{\x}} - \frac{\x+\u}{\twonorm{\x+\u}} }\\
& \le \frac{\delta}{2} + \twonorm{ \inparen{\frac{1}{\twonorm{\x}} - \frac{1}{\twonorm{\x+\u}} }\x} + \twonorm{\frac{\u}{\twonorm{\x+\u}}}\\
& \le \frac{\delta}{2} + \frac{ | \twonorm{\x+\u} - \twonorm{\x} |}{\twonorm{\x+\u}} + \frac{\twonorm{\u}}{\twonorm{\x+\u}}
\le \frac{\delta}{2} + \frac{2 \twonorm{\u}}{\twonorm{\x+\u}}.
\end{align*}
From $\twonorm{\u} \le d \twonorm{\e} / \sqrt{q}$
and  
$\twonorm{\x+\u} 
\ge (1-dc_6) \twonorm{\x} \ge \twonorm{\x}/2$ for $c_6$ is small enough,
we derive that
\begin{equation} \label{StepRobHT3}
\twonorm{ \frac{\x}{\twonorm{\x}} - \frac{\sqrt{\pi/2}}{q}  \matA_{S \cup T}^* \inparen{\sgn \inparen{ \matA \x + \e }}} 
\le \frac{\delta}{2} + \frac{4d \twonorm{\e}}{\sqrt{q}\twonorm{\x}}.
\end{equation}
Substituting \eqref{StepRobHT2} and \eqref{StepRobHT3} into \eqref{StepRobHT1} enables us to derive the desired result \eqref{ObjRobHT2} from \eqref{StepRobHT0}.
\end{proof}

The proof of Theorem~\ref{thm:hardthresh} presented next follows from Theorem~\ref{ThmRobHT}.

\begin{proof}[Proof of Theorem~\ref{thm:hardthresh}]
For later purposes, we introduce the constant
$$
C:= \max_{\xi \in \left[\frac{1}{\sqrt{2}}-\frac{1}{20},\frac{2}{\sqrt{5}}+\frac{1}{20} \right] } \left| f'(\xi) \right|
\ge 2,
\qquad
f(\xi):= 1-\frac{\sqrt{1-\xi^2}}{\xi}.
$$
Given $\x \in R B_2^n \cap \Sigma_s$,
we acquire a corrupted version $\y_1 \in \{\pm 1\}^{q/2}$ of the quantized measurements $\sign( \matA_1 \x)$.
Since the number of rows of the matrix $\matA_1 \in \R^{(q/2) \times n}$ is large enough for
Theorem \ref{ThmRobHT} to hold with $\delta_0 = \delta/(4(1+2C))$ instead of $\delta$,
we obtain
$$
\twonorm{\frac{\x}{\twonorm{\x}} - \u}
\le \delta_0 + c_4 c \delta + c_5 c' \delta \le 2 \delta_0,
\qquad 
\u := H_s'(\matA_1^* \y_1),
$$
provided that the constants $c$ and $c'$ are small enough.
With $\x^\sharp$ denoting the orthogonal projection of $\x$ onto the line spanned by $\u$, we have
$$
\twonorm{\x-\x^\sharp} \le \twonorm{ \x - \twonorm{\x} \u } \le 2 \delta_0 \twonorm{\x}.
$$
We now consider a unit-norm vector $\v \in \R^n$ supported on $\supp(\u)$ and orthogonal to $\u$.
The situation in the plane spanned by $\u$ and $\v$ is summarized in Figure~\ref{fig:situation}.

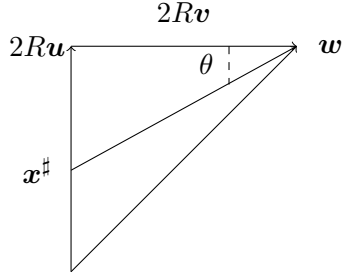
\begin{figure}
\begin{center}
\begin{tikzpicture}[scale=3]
\draw[->] (0,0) -- (0,1);
\draw[->] (0,1) -- (1,1);
\draw[->] (0,0) -- (1,1);
\draw (1,1) -- (0,.45);
\node at (-.15, .45) {$\x^\sharp$};
\node at (-.15, 1) {$2R\u$};
\node at (.5, 1.15) {$2R\v$};
\node at (1.15, 1) {$\w$};
\draw[dashed] (.7, 1) -- (.7,.83);
\node at (.6, .92) {$\theta$};
\end{tikzpicture}
\end{center}
\caption{The situation in the plane spanned by $\u$ and $\v$.}
\label{fig:situation}
\end{figure}

We point out that $\|\x^\sharp\| \le \|\x\| \le R$ gave $\|\x^\sharp\|_2 \le 2R$,
but that $2R$ was just an arbitrary choice to ensure that $\cos(\theta)$ stays away from $1$---here, $\cos(\theta) \in [1/\sqrt{2}, 2/\sqrt{5}]$.
Forming the $s$-sparse vector $\w:=2R \cdot (\u+\v)$, 
we now acquire a corrupted version $\y_2 \in \{\pm 1\}^{q/2}$ of the quantized measurements $\sign( \matA_2 (\x- \w))$
on the $2s$-sparse vector $\x - \w$.
Since the number of rows of the matrix $\matA_2 \in \R^{(q/2) \times n}$ is large enough for
Theorem \ref{ThmRobHT} to hold with $\delta_0 = \delta/(4(1+2C))$ instead of $\delta$ and $2s$ instead of $s$,
we obtain
$$
\twonorm{ \frac{\w - \x}{\twonorm{\w - \x}} - \t }
\le \delta_0 + c_4 c \delta + c_5 c' \delta \le 2 \delta_0,
\qquad
\t = - H_s'(\matA_2^* \y_2).
$$
We deduce that $\t$ also approximates $(\w - \x^\sharp)/\twonorm{\w - \x^\sharp}$ with error 
\begin{align*}
\Bigg\| & \frac{\w- \x^\sharp}{\|\w - \x^\sharp\|} - \t \Bigg\|_2\\
& \le
\twonorm{
\frac{\w-\x^\sharp}{\twonorm{\w-\x^\sharp}} - \frac{\w-\x}{\twonorm{\w-\x^\sharp}}
}
+
\twonorm{ 
\left( \frac{1}{\twonorm{\w - \x^\sharp}} - \frac{1}{\twonorm{\w-\x}} \right) (\w-\x)
}
+\twonorm{
\frac{\w-\x}{\twonorm{\w-\x}} - \t
} \\
& \le 
\frac{\twonorm{\x-\x^\sharp}}{\twonorm{\w-\x^\sharp}} + \frac{\left| \big\| \w-\x \big\|_2 - \twonorm{\w-\x^\sharp} \right|}{\twonorm{\w-\x^\sharp}}
+ 2 \delta_0 
 \le 2 \frac{\twonorm{\x-\x^\sharp}}{\twonorm{\w-\x^\sharp}}  + 2 \delta_0
\le 2 \frac{2 \delta_0 \twonorm{\x}}{2R} +  2 \delta_0\\
&  \le 4 \delta_0.
\end{align*}
It follows that $\ip{\t}{\v}$  approximates $ \ip{(\w-\x^\sharp)/\|\w-\x^\sharp\|}{\v} = \cos(\theta)$ with error
$$
|\cos(\theta) - \ip{\t}{\v}| 
= \left| \ip{\frac{\w-\x^\sharp}{\twonorm{\w-\x^\sharp}} - \t}{\v} \right|
\le \twonorm{ \frac{\w - \x^\sharp}{\twonorm{\w-\x^\sharp}} - \t } \twonorm{\v} \le 4 \delta_0.
$$
We then notice that 
$$
\twonorm{\x^\sharp} = 2R - 2R \tan(\theta) = 2R f(\cos(\theta)),
$$
so that $2R f(\ip{\t}{\v})$ approximates $\twonorm{\x^\sharp}$ with error
$$
\left| \twonorm{\x^\sharp} -  2R f(\ip{\t}{\v}) \right|
= 2R |f(\cos(\theta)) - f(\ip{\t}{\v})|
\le 2R \, C \, |\cos(\theta) - \ip{\t}{\v}| 
\le 2R \, C \, 4 \delta_0 = 8 C \delta_0 R.
$$
Here, we used the facts that $\cos(\theta) \in [1/\sqrt{2}, 2/\sqrt{5}]$
and that $\ip{\t}{\v} \in [1/\sqrt{2}-4 \delta_0, 2/\sqrt{5}+4\delta_0] \subseteq [1/\sqrt{2}-1/20, 2/\sqrt{5}+1/20]$.
We derive that 
\begin{align*}
\left| \Big\| \x \Big\|_2 - 2R f(\ip{\t}{\v}) \right|
& \le \left| \Big\| \x \Big\|_2 - \twonorm{\x^\sharp } \right| + \left| \twonorm{\x^\sharp} -  2 R f(\ip{\t}{\v}) \right|\\
& \le \twonorm{\x-\x^\sharp} + \left|  \twonorm{\x^\sharp} -  2 R f(\ip{\t}{\v}) \right|\\
& \le 2 \delta_0 \twonorm{\x} + 8 C \, \delta_0 R 
  \le 2(1+4C) \delta_0 R.
\end{align*}
Finally, with the estimate $\hat{\x}$ for $\x$ being defined as
$$
\hat{\x} := 2 R f(\ip{\t}{\v}) \, \u,
$$
the previous considerations lead to the error estimate
\begin{align*}
\twonorm{\x-\hat{\x}}
& \le \twonorm{\x - \twonorm{\x} \u }  + \left| \twonorm{\x} - 2 R f(\ip{\t}{\v}) \right| \twonorm{\u}
\le 2 \delta_0 \twonorm{\x} + 2(1+4C) \delta_0 R  \\
& \le 4(1+2C) \delta_0 R .
\end{align*}
Our initial choice of $\delta_0 = \delta/(4(1+2C))$ enables us to conclude that $\twonorm{\x-\hat{\x}} \le \delta R$.
\end{proof}

\subsection{Second-order-cone-programming-based order-one recovery scheme}

\begin{proof}[Proof of Theorem~\ref{thm:linpro}]
Without loss of generality, we assume that $R = 1/2$.
The general argument follows from a rescaling.  
We begin by considering the exact case in which $\e = {\bf 0}$.  
Observe that, by the Cauchy--Schwarz inequality,
\[\onenorm{\x} \leq \sqrt{\zeronorm{\x}} \cdot \twonorm{\x} \leq \sqrt{s}.\]
Since $\x$ is feasible for program \eqref{eq:socp}, we also have $\onenorm{\hat{\x}} \leq \sqrt{s}$.  
The result will follow from the following two observations:
\begin{itemize}
\item $\x, \hat{\x} \in \sqrt{s} B_1^n \cap B_2^n$
\item  $\sign(\< \a_i, \x\> - \tau_i  ) = \sign(\< \a_i, \hat{\x}\> - \tau_i), \qquad i = 1, \hdots, q$.
\end{itemize}
Each equation $\< \a_i, \z \> - \tau_i = 0$ defines a hyperplane perpendicular to $\a_i$ and translated proportionally to $\tau_i$; further, $\x$ and $\hat{\x}$ are on the same side of the hyperplane.  
To visualize this, imagine  $\sqrt{s} B_1^n \cap B_2^n$ as an oddly shaped apple that we are trying to dice.  
Each hyperplane randomly slices the apple, eventually cutting it into small sections.  
The vectors $\hat{\x}$ and $\x$ belong to the same section.  
Thus, we ask: \textit{how many random slices are needed for all sections to have small diameter?}  
Similar questions have been addressed in a broad context in \cite{pv-embeddings}.  We give a self-contained proof that $O(s \log(n/s))$ slices suffice based on the following result \cite[Theorem 3.1]{pv-embeddings}.

\begin{theorem}[Random hyperplane tessellations of $\sqrt{s} B_1^n \cap S^{n-1}$]
\label{thm:embeddings}
Let $\a_1, \a_2, \hdots, \a_q \in \R^n$ be independent standard normal vectors.  
If
\[q \geq C \delta^{-4} s \log(n/s),\]
then, with probability at least $1 - 2 \exp(-c \delta^4 q)$,
 all $\x, \x' \in \sqrt{s} B_1^n \cap S^{n-1}$ with
\[\sign \< \a_i, \x\>  = \sign \< \a_i, \x'\>  , \qquad i = 1, \hdots, q,\]
satisfy
\[\twonorm{\x - \x'} \leq \frac{\delta}{8}.
\]
The positive constants $c$ and $C$ are absolute constants.
\end{theorem}

We  translate the above result into a tessellation of $\sqrt{s} B_1^n \cap B_2^n$ in the following corollary.  

\begin{cor}[Random hyperplane tessellations of $\sqrt{s} B_1^n \cap B_2^n$]
\label{cor:tessellations}
Let $\a_1, \a_2, \hdots, \a_q \in \R^n$ be independent standard normal vectors and let $\tau_1, \tau_2, \hdots, \tau_q$ be independent standard normal random variables.  
If
\[q \geq C \delta^{-4} s \log(n/s),\]
then, with probability at least $1 - 2 \exp(-c \delta^4 q)$,
 all $\x, \x' \in \sqrt{s} B_1^n \cap B_2^n$ with
\[\sign(\< \a_i, \x\> - \tau_i) = \sign(\< \a_i, \x'\> - \tau_i), \qquad i = 1, \hdots, q,\]
satisfy
\[\twonorm{\x - \x'} \leq \frac{\delta}{4}.
\]
The positive constants $c$ and $C$ are absolute constants.
\end{cor}

\begin{proof}
For any $\z \in \sqrt{s} B_1^n \cap B_2^n$,
we notice that $\sign( \< \a_i, \z \> - \tau_i) = \sign( \< [\a_i, -\tau_i], [\z, 1] \> )$,
where the augmented vectors $[\a_i, -\tau_i] \in \R^{n+1}$ and $[\z, 1] \in \R^{n+1}$ are the concatenations of $\a_i$ with $-\tau_i$ and $\z$ with $1$, respectively.
Thus, we have moved to the ditherless setup by only increasing the dimension by one.
Since
$$
\twonorm{[\z,1]} \ge 1
\quad \mbox{and} \quad
\onenorm{[\z,1]} = \onenorm{\z} + 1 \le \sqrt{s}+1 \le \sqrt{4s},
$$
we may apply Theorem \ref{thm:embeddings} after projecting on $S^n$ to derive 
\begin{equation}
\label{eq:original equation}
\twonorm{\frac{[\x, 1]}{\twonorm{[\x, 1]}} -  \frac{[\x', 1]}{\twonorm{[\x', 1]}}} \leq \frac{\delta}{8}.
\end{equation}
with probability at least $1 - 2 \exp(c \delta^4 q)$.
We now show that the inequality (\ref{eq:original equation}) implies that $\twonorm{\x - \x'} \leq \delta/4$.  

First note that
\[
\twonorm{\x - \x'} \leq \sqrt{2} \, \twonorm{\frac{\x}{\twonorm{[\x, 1]}} - \frac{\x'}{\twonorm{[\x, 1]}}}
\]
since $\twonorm{\x} \leq 1$.  
Subtract and add $\x'/\twonorm{[\x', 1]}$ inside the norm and apply triangle inequality to obtain
\[
\twonorm{\x - \x'} \leq \sqrt{2}\left( \twonorm{\frac{\x}{\twonorm{[\x, 1]}} - \frac{\x'}{\twonorm{[\x', 1]}}} + \twonorm{\x'}\cdot \left |\frac{1}{\twonorm{[\x, 1]}} - \frac{1}{\twonorm{[\x', 1]}} \right| \right).
\]
Since $\twonorm{\x'} \leq 1$, we may remove $\twonorm{\x'}$ from in front of the second term in parenthesis.  
Next, use the inequality $a + b \leq \sqrt{2} \cdot \sqrt{a^2 + b^2}$ on the two terms in parenthesis.  This bounds the right-hand side by precisely 
\[2 \twonorm{\frac{[\x, 1]}{\twonorm{[\x, 1]}} -  \frac{[\x', 1]}{\twonorm{[\x', 1]}}},\]
which is bounded by $\delta/4$ according to \eqref{eq:original equation}.
\end{proof}

This corollary immediately completes the proof of Theorem~\ref{thm:linpro} in the case $\e  = {\bf 0}$.  
We now turn to the general problem where $\inftynorm{\e} \leq c \delta^3$ and thus $\twonorm{\e} \leq c \delta^3 \sqrt{q}$.  
We reduce to the exact problem using the simultaneous $(\ell_1,\ell_2)$-quotient property \eqref{SimQP}, 
which guarantees that the error can be represented by a signal with small $\ell_1$-norm.  
In particular, \eqref{SimQP} implies that,
with probability at least $1 - \exp(-c q)$,
there exists a vector $\u$ satisfying 
\begin{equation}
\e  = \matA \u
\qquad \mbox{with} \quad
\left\{
\begin{matrix}
\twonorm{\u} & \le & \delta/4,\\
\onenorm{\u} & \le & c_1 \delta^3 \sqrt{q/\log(n/q) } 
\end{matrix}
\right.
\end{equation}
where $c_1$ is an absolute constant which we may choose as small as we need.
We may now replace $\x$ with $\tilde{\x} = \x + \u$ and proceed as in the proof in the noiseless case.  
Reconstruction of $\tilde{\x}$ to accuracy $\delta/4$ yields reconstruction of $\x$ to accuracy $\delta/2$, as desired.  
By replacing $\x$ with $\tilde{\x}$, we have (mildly) increased the bound on the $\ell_1$-norm and the $\ell_2$-norm.  Fortunately, $\twonorm{\tilde{\x}} \leq \twonorm{\x} + \twonorm{\u} \leq 1$ and thus $\tilde{\x}$ remains feasible for the program \eqref{eq:socp}.  Further, $\tilde{\x}$ is approximately sparse in the sense that $\onenorm{\tilde{\x}} \leq \onenorm{\x} + \onenorm{\u} \leq \sqrt{s} + c_1 \delta^3 \sqrt{q/\log(n/q)} =: \sqrt{\tilde{s}}$.  
To conclude the proof, we must show that the requirement of Theorem \ref{thm:linpro}, namely $q \geq C' \delta^{-4} s \log(n/s)$,
implies that the required condition of Corollary \ref{cor:tessellations}, 
namely $q \geq C \delta^{-4} \tilde{s} \log(n/\tilde{s})$,
is still satisfied.  The result follows from massaging the equations, as sketched below.

If $s \geq c_1^2 \delta^6 q/\log(n/q)$, then $\sqrt{\tilde{s}} \leq 2 \sqrt{s}$ and the desired result follows quickly.  Suppose then that $s < c_1^2 \delta^6 q/\log(n/q)$ and thus $\tilde{s} \leq c_2 \delta^6 q/\log(n/q)$.  To conclude, note that
\[C \delta^{-4} \tilde{s} \log(n/\tilde{s}) \leq q \cdot C \cdot c_2\frac{\delta^2}{\log(n/q)}\cdot (\log(n/q) + \log(1/c_2) + 6 \log(1/\delta) + \log(\log(n/q)) \leq q,\]
where the first inequality follows since $s \log(n/s)$ is increasing in $s$ and thus $\tilde{s}$ may be replaced by its upper bound, $c_2 \delta^6 q/\log(n/q)$.  The last inequality follows by taking $c_2$ small enough. This concludes the proof.
\end{proof}

\section{Numerical Results}\label{sec:exps}

This brief section provides several experimental validations of the theory developed above.
The computations,
performed in MATLAB,
are reproducible and can be downloaded from the second author's webpage. 
The random measurements $\a_i$ were always generated as vectors with independent standard normal entries.
As for the random sparse vectors $\x$,
after a random choice of their supports,
their nonzero entries also consisted of independent standard normal variables.

Our first experiment (results not displayed here) verified on a single sparse vector that both its direction and magnitude can be accurately estimated via order-one recovery schemes,
while only its direction could be accurately estimated using convex programs \cite{pv-1-bit,pv-noisy-1bit}, $\ell_1$-regularized logistic regression, or binary iterative hard thresholding \cite{Jacques2011}.
We also noted the reduction of the reconstruction error by several orders of magnitude 
from the same number $m$ of quantized measurements
when Algorithms~\ref{alg:Q}-\ref{alg:Delta} are used instead of the above methods.
We remark in passing that this number $m$ is significantly larger than the number of measurements in classical compressed sensing with real-valued measurements, as intuitively expected.

Our second experiment corroborates the exponential decay of the error rate.
The results are summarized in Figure \ref{FigExpDecay},
whose logarithmic scale on the vertical axis confirms the behavior 
$\log(\|\x-\x^*\|_2/\|\x\|_2) \le - c \lambda$
for the relative reconstruction error as a function of the oversampling factor $\lambda=m / \log(n/s)$.
The tests were conducted on four sparsity levels $s$ at a fixed dimension $n$
for an oversampling ratio $\lambda$ varying through the increase of the number $m$ of measurements.
The number $T$ of iterations in Algorithms \ref{alg:Q} and \ref{alg:Delta} was fixed throughout the experiment based on hard thresholding
and throughout the experiment based on second-order cone programming.
The values of all these parameters are reported directly in Figure \ref{FigExpDecay}.
We point out that we could carry out a more exhaustive experiment for the faster hard-thresholding-based version than for the slower second-order-cone-programming-based version,
both in terms of problem scale and of number of tests.

\begin{figure}[htbp]
\center
\subfigure{
(a) \includegraphics[width=0.42\textwidth]{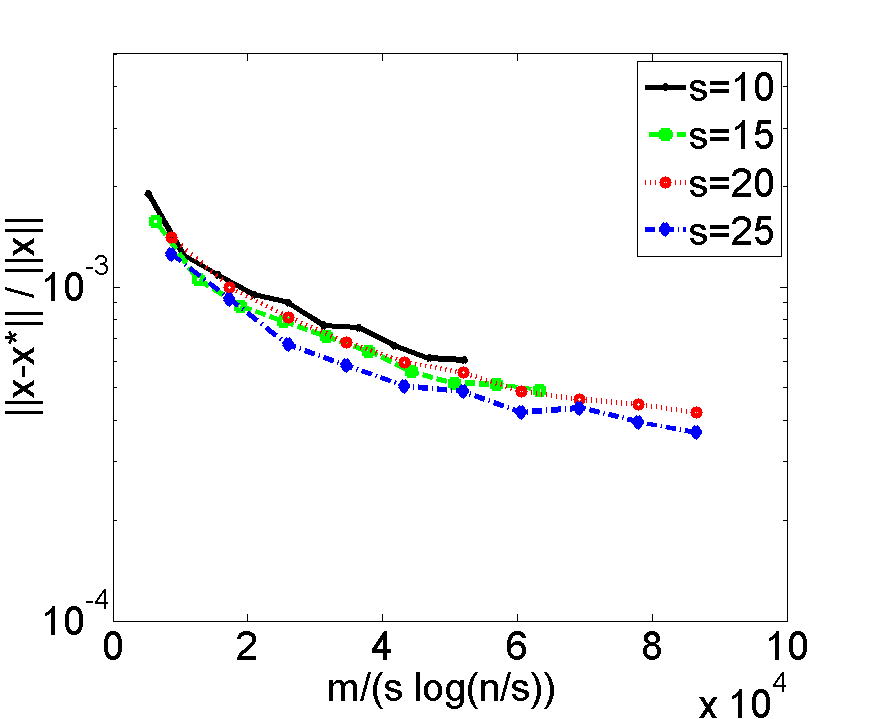}
\label{FigExpDecayHT}
}\quad
\subfigure{
(b) \includegraphics[width=0.42\textwidth]{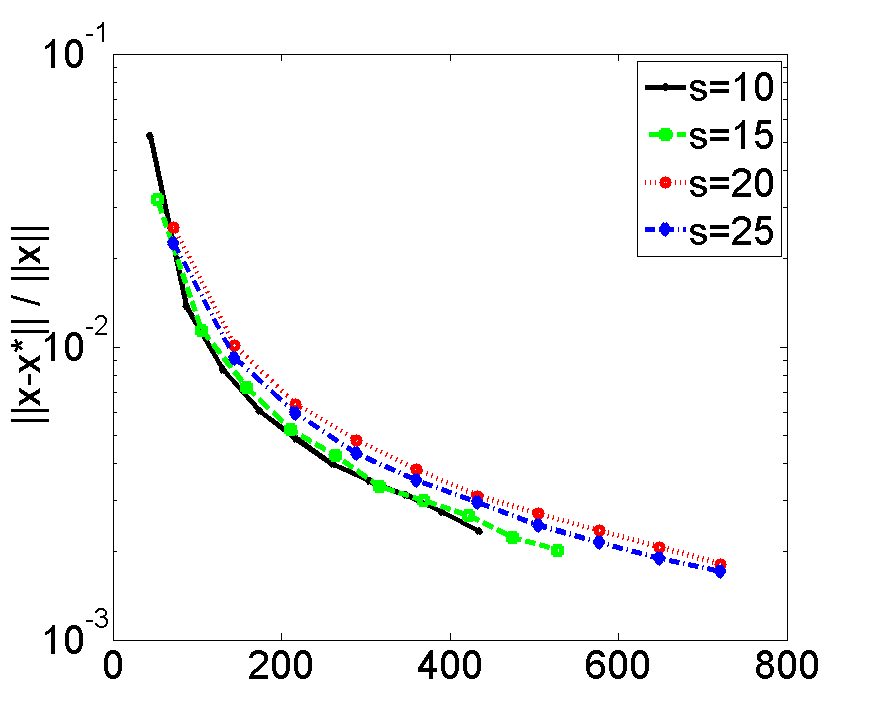}
\label{FigExpDecayLP}
}
\caption{Averaged relative error for the reconstruction of sparse vectors ($n = 100$) by the outputs
of Algorithms~\ref{alg:Q}-\ref{alg:Delta} based on (a) hard thresholding  and (b) second-order cone programming 
as a function of the oversampling ratio.
}
\label{FigExpDecay}
\end{figure}

Our third experiment examines the effect of measurement errors on the reconstruction via Algorithms \ref{alg:Q} and \ref{alg:Delta}.
Once again,
the problem scale was much larger when relying on hard thresholding than on second-order cone programming.
The values of the size parameters are reported on Figure \ref{FigNoise}.
This figure shows how the reconstruction error decreases as the iteration count $t$ increases in Algorithms \ref{alg:Q} and \ref{alg:Delta}.
For the hard-thresholding-based version,
see Figure \ref{FigNoise}(a),
we observe an error decreasing by a constant factor at each iteration when the measurements are totally accurate.
Introducing a pre-quantization noise $\e \sim N(0,\sigma^2 {\bf I})$
in $\y = \sign (\matA \x + \e)$ does not affect this behavior too much
until the ``noise floor'' is reached.
Flipping a small fraction of the bits $\sgn \ip{\a_i}{\x}$
by multiplying them with $f_i = \pm 1$,  most of which being equal to $+1$,
seems to have an even smaller effect on the reconstruction.
However, these bit flips prevent the use of the second-order-cone-programming-based version,
as the constraints of the optimization problems become infeasible.
But we still remark that the pre-quantization noise is not very damaging in this case either,
see Figure \ref{FigNoise}(b),
where the results of an experiment using $\ell_1$-regularized logistic regression in Algorithms \ref{alg:Q} and \ref{alg:Delta} are also displayed.

\begin{figure}[htbp]
\center
\subfigure{
(a) \includegraphics[width=0.42\textwidth]{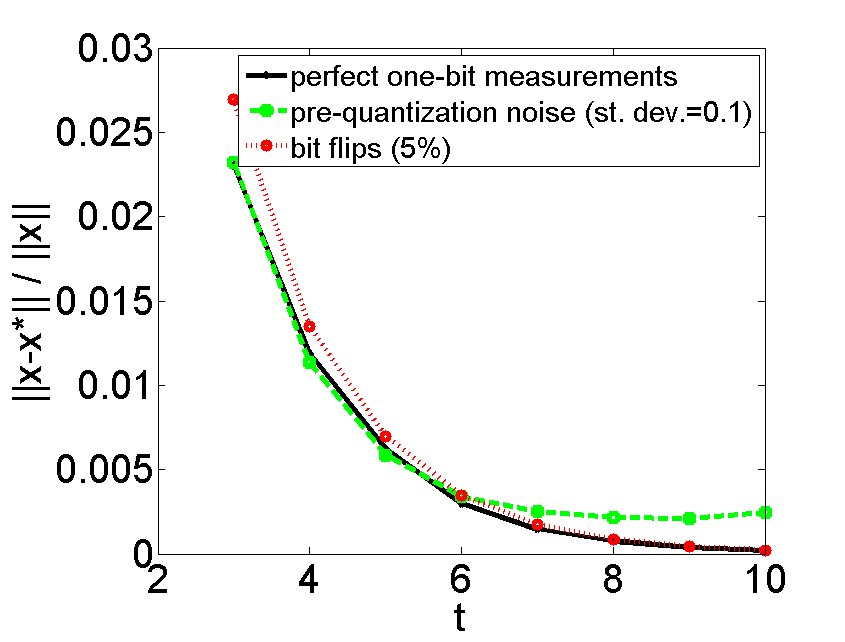}
\label{FigNoiseHT}
}\quad
\subfigure{
(b) \includegraphics[width=0.42\textwidth]{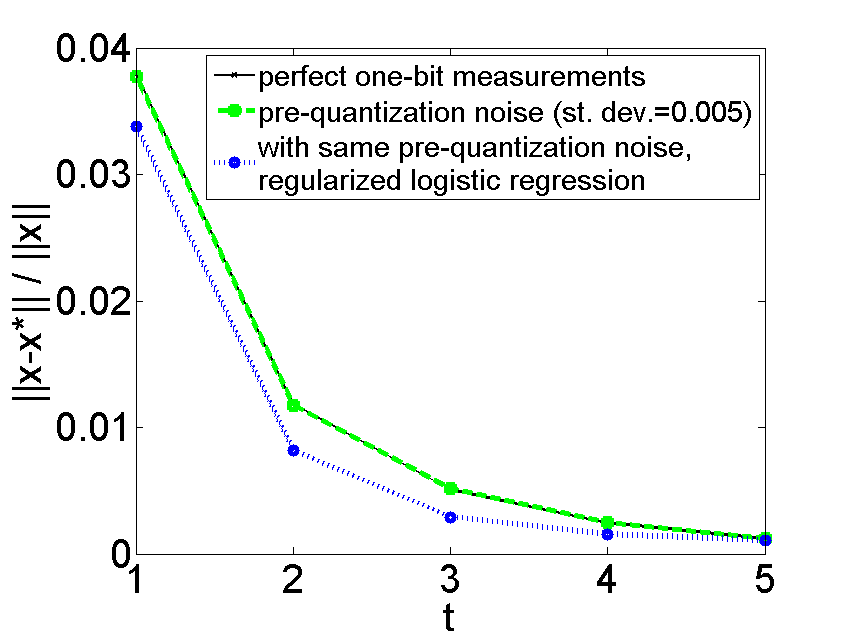}
\label{FigNoiseLP}
}
\caption{Averaged relative error for the reconstruction of sparse vectors ($n=100$) by the outputs
of Algorithms~\ref{alg:Q}-\ref{alg:Delta} based on (a) hard thresholding ($s=15$, $m=10^5$)  and second-order cone programming and (b) $\ell_1$-regularized logistic regression ($s=10$, $m=2\cdot 10^4$)
as a function of the iteration count when measurement error is present.
}
\label{FigNoise}
\end{figure}

\section{Discussion}\label{sec:discussion}

\subsection{Related work}\label{sec:related}

The one-bit compressed sensing framework developed by Boufounos and Baraniuk~\cite{Boufounos2008} is a relatively new line of work, with theoretical backing only recently being developed.  Empirical evidence and convergence analysis of algorithms for quantized measurements appear in the works of Boufounos et al. and others~\cite{Boufounos2009,Boufounos2008,Laska2010,Zymnis2010}.  Theoretical bounds on recovery error have only recently been studied, outside from results which model the one-bit setting as classical compressed sensing with specialized additive measurement error~\cite{Dai2009,jacques2011dequantizing,Sun2009}.  Other settings analyze quantized measurements where the number of bits used depends on signal parameters like sparsity level or the dynamic range~\cite{ardestanizadeh2009,gunturk2010,gunturk2010paper}.  Boufounos develops hierarchical and scalar quantization with modified quantization regions which aim to balance the rate-distortion trade-off~\cite{boufounos2011hierarchical,boufounos2012universal}.  These results motivate our work but do not directly apply to the compressed sensing setting.

Theoretical guarantees more in line with the objectives of this paper began with Jacques et al.~\cite{Jacques2011} who proved robust recovery from approximately $s\log n$ one-bit measurements.  However, the program used has constraints which require sparsity estimation, making it NP-Hard in general.  Gupta et al. offers a computationally feasible method via a scheme which either depends on the dynamic range of the signal or is adaptive~\cite{Gupta2010}.  Plan and Vershynin analyze a tractable non-adaptive convex program which provides accurate recovery without these types of dependencies~\cite{pv-1-bit,pv-noisy-1bit,alpv-1bit}.  Other methods have also been proposed, many of which are largely motivated by classical compressed sensing methods (see e.g.~\cite{Boufounos2009,movahed2012robust,yan2012robust,ma2013two,jacques2013quantized}).  

In order to break the bound \eqref{eq:bestknown} and obtain an exponential rather than polynomial dependence on the oversampling factor, one cannot take traditional non-adaptive measurements.  Several schemes have employed adaptive samples including the work of Kamilov et. al. which utilizes a generalized approximate message passing algorithm (GAMP) for recovery, and the adaptive thresholds are selected in line with this recovery method.  Adaptivity is also considered in~\cite{Gupta2010} which allows for a constant factor improvement in the number of measurements required.  However, to our best knowledge our work is the first to break the bound given by~\eqref{eq:bestknown}.

Regarding the link between our methods and sparse binary regression, there is a number of related theoretical results focusing on sparse logistic regression \cite{negahban2010unified, bunea2008honest,van2008high, bach2010self, ravikumar2010high, meier2008group, kakade2009learning}, but these are necessarily constrained by the same limited accuracy of the one-bit compressed sensing model discussed in Section \ref{sec:intro}.  

We also point to the closely related threshold group testing literature, see e.g., \cite{cheraghchi2013improved}.  In many cases, the statistician has some control over the threshold beyond which the measurement maps to a one.  For example, the wording of a binary survey may be adjusted to only ask for a positive answer in an extreme case; a study of the relationship of heart attacks to various factors may test whether certain subjects have heart attacks in a short window of time and other subjects have heart attacks in a long window of time.  The main message of this paper is that by carefully choosing this threshold the accuracy of reconstruction of the parameter vector $\x$ can be greatly increased.

\subsection{Conclusions}

We have proposed a recursive framework for adaptive thresholding quantization in the setting of compressed sensing.  
We have developed both a second-order-cone-programming-based method and a hard-thresholding-based method for signal recovery from these type of quantized measurements.  
Both of our methods feature a bound on the recovery error of the form $e^{-\Omega(\lambda)}$, an exponential dependence on the oversampling factor $\lambda$.  To our best knowledge, this is the first result of this kind, and it improves upon the best possible dependence of $\Omega(1/\lambda)$ for non-adaptively quantized measurements.  

\subsection*{Acknowledgements}
We would like to thank the AIM SQuaRE program for hosting our initial collaboration and also Mr.\ Lan for discussions around the relationship of our work to logistic regression.

\bibliographystyle{myalpha}
\bibliography{one-bit-bib}
\end{document}